\documentclass[atmp]{ipart_v1}

\firstpage{1}

\usepackage{t1enc}
\usepackage[latin1]{inputenc}
\usepackage[english]{babel}

\usepackage{amsthm}
\usepackage{yfonts}

\usepackage{bbm}
\usepackage{bm}
\usepackage{mathrsfs}
\usepackage{booktabs}
\usepackage{float}

\newcommand{\be}[0]{\begin{equation}}
\newcommand{\ee}[0]{\end{equation}}

\numberwithin{equation}{section}

\theoremstyle{plain}
\newtheorem{Theorem}{Theorem}[section]

\newtheorem{Proposition}{Proposition}[section]
\newtheorem{Definition}{Definition}[section]
\newtheorem{Corollary}{Corollary}[section]

\begin{document}

\title[The Zak phase in topologically insulating chains]{The Zak phase in topologically insulating chains: invariants and limitations}

\author[Federico Manzoni, Domenico Monaco, Gabriele Peluso]{Federico Manzoni, Domenico Monaco, Gabriele Peluso}

\begin{abstract}
In this work we investigate the topological content of the Zak phase in one-dimensional translation-invariant topological insulators endowed with time-reversal, particle-hole and/or chiral symmetries, extending results from \cite{Monaco_2023}. We analyze the extent to which the Zak phase captures the topology of all Altland--Zirnbauer--Cartan (AZC) symmetry classes in $1$D. Building on the framework of fibered Hamiltonians and spectral projections, we construct symmetric Bloch bases adapted to the discrete symmetries of the system and define a $\mathbb{Z}_2$-valued topological invariant  $\mathrm{I}^{(\mathrm{AZC-class})}(H)$ obtained from the abelian Zak phase. Moreover, we demonstrate that in symmetry classes admitting a quaternionic structure, i.e. anti-unitary symmetries squaring to minus the identity, the Zak phase is further constrained, leading to the vanishing of the $\mathbb{Z}_2$ invariant mentioned above. This highlights the sensitivity of the Zak phase to additional geometric structures of the manifold of occupied energy states, as well as its limitations in being an effective marker for topological phases of insulating chains. As an example, we discuss the case of generalized Kitaev chains with arbitrary finite-range hopping and single or multiple chiral channels, and show how the Zak phase only retains partial information about their different topological phases.
\end{abstract}

\maketitle

\section{Introduction}
In recent decades, a new paradigm in condensed matter physics has emerged with the discovery of topological phases of matter~\cite{chiu2016classification}. Unlike conventional phases characterized by local order parameters and spontaneous symmetry breaking, topological phases are defined by global, topological invariants that remain robust under continuous deformations of the system. These phases exhibit quantized physical responses and host exotic surface states protected by topological features of the bulk system. The study of topological matter has deepened the connection between quantum physics and abstract mathematical frameworks, especially topology and $K$-theory, leading to both theoretical advances and technological implications. We briefly review this setting in Section~\ref{4.1}.

In order to apply the results from the abstract classification of topological phases of matter in concrete physical models, one needs to find appropriate topological markers, namely indices which are able to differentiate (and possibly identify uniquely) the topological class of a given microscopic state. In this paper, we address the latter question in one-dimensional (1D) translation-invariant topological insulators. One commonly adopted topological marker used to probe the properties of insulating quantum chains is the Zak phase~\cite{zak1989berry}, which for an isolated energy band computes the geometric Berry phase picked up by a Bloch function when it is transported along the 1D Brillouin torus (cf.~Definition~\ref{def:ZakPhase} below). The Zak phase underlies the computation of polarization within the modern theory of electronic structure in solids~\cite{xiao2010berry, resta1994}; it is experimentally measurable in quantum simulators made up of ultracold atoms in optical lattices \cite{atala2013direct}; its widespread use in condensed matter physics is also due to its computability in terms of Wannier functions and their centers of localization, which are easily accessible in the numerical simulation of real materials via software like \texttt{wannier90} \cite{mostofi2008wannier90, marrazzo2024wannier}. 

Given the relevance of the Zak phase for topological transport in periodic media, we systematically investigate this marker in 1D multi-band models in all the Altland--Zirnbauer--Cartan (AZC) symmetry classes, as depicted in Kitaev's ``periodic table''~\cite{Kitaev_2009, Ryu_2010} (see Section~\ref{4.1} for a review of the concepts and notation used here). The main questions we address are the following:
\begin{quote}
Can the Zak phase be used to detect and classify topological phases? If so, how complete is the topological information accessed by the Zak phase? How sensitive is the Zak phase to the presence of symmetries or other structures in the model?    
\end{quote}
We show that, in general, the Zak phase associated to negative energy bands allows to define a gauge and topological invariant $\mathrm{I}^{(\mathrm{AZC-class})}(H)$ with values in $\mathbb{Z}_2$: this extends the arguments and findings of~\cite{Monaco_2023} to all symmetry classes. In presence of a quaternionic structure induced by a anti-unitary symmetry squaring to minus the identity, we show how the invariant extracted from the Zak phase necessarily vanishes. For a comparison of the value of the proposed $\mathbb{Z}_2$-valued invariant with the prediction from Kitaev's table, which is instead based on $K$-theory, see Table \ref{tabclassinv2}.

\begin{table}[h!]
\centering
\caption{Comparison between $K$-theoretic invariants and $\mathrm{I}^{(\mathrm{AZC-class})}$ for $1$D AZC classes.}\label{tabclassinv2}
\begin{tabular}{@{}l|ccc|ccccccccc@{}}
\toprule
\textbf{Symmetry Class} & $T$ & $C$ & $S$ & $K$-theory & $\mathrm{I}^{(\mathrm{AZC-class})}$ \\
\midrule
A     & 0 & 0 & 0 & 0 & $0$ \\
AIII  & 0 & 0 & 1 & $\mathbb{Z}$ & $\mathbb{Z}_2$ \\
\midrule
AI    & 1 & 0 & 0 & 0 & 0  \\
BDI   & 1 & 1 & 1 & $\mathbb{Z}$ & $\mathbb{Z}_2$ \\
D     & 0 & 1 & 0 & $\mathbb{Z}_2$ & $\mathbb{Z}_2$ \\
DIII  & -1 & 1 & 1 & $\mathbb{Z}_2$ & 0 \\
AII   & -1 & 0 & 0 & 0 & $0$ \\
CII   & -1 & -1 & 1 & $\mathbb{Z}$ & 0 \\
C     & 0 & -1 & 0 & 0 & 0 \\
CI    & 1 & -1 & 1 & 0 & 0 \\
\bottomrule
\end{tabular}
\end{table}

Our findings show that, despite its important applications in condensed matter physics presented above, the Zak phase yields in general a limited topological marker in the broader context of 1D topological insulators. Topological phases of matter, both in presence and absence of lattice-translation invariance, are classified by more refined indices, which can be defined in all dimensions via functional analysis and Fredholm indices~\cite{grossmann2016index}, via $K$-theory~\cite{Kitaev_2009, freed2013twisted, thiang2016k}, or via homotopy theory~\cite{kennedy2015homotopy, chung2025topological, Chung_Shapiro_II, Chung_Shapiro_All, PelusoThesis, SantiThesis}.

The paper has the following structure. In Section \ref{4.1} we introduce our setup: topological insulators and possible model Hamiltonians for them, discrete symmetry operators, and the AZC classification. This discussion is complemented by the Appendices. We also study the action of the discrete symmetry operators (and their combinations) on the Hamiltonian and the explicit writing of the fiber Hamiltonians and spectral eigenprojections, obtained by exploiting lattice-translation invariance. In Section \ref{BB} we give the general definition of Bloch basis for the spectral eigenprojection of the Hamiltonian on the relevant energy bands, and identify the conditions which make such a basis compatible with the symmetries of the model. We also show how a symmetric Bloch basis can be explicitly constructed through parallel transport. These first Sections review material presented in \cite{Monaco_2023}, which focused on symmetry classes with chiral or particle-hole symmetries, and extend the results to cover the whole ``periodic table''. In Section \ref{4.3} we define the Zak phase and we extract the topological $\mathbb{Z}_2$-valued invariant I$^{\text{(AZC--class)}}(H)$. In Section \ref{quad0} we shed light on the interplay between the invariant and the presence of a quaternionic structure related to an anti-linear symmetry operator that squares to minus the identity. In Section \ref{focBDI} we make explicit examples, considering generalized Kitaev chains with arbitrary finite-range hopping and possibly multiple chiral channels, where we show how the invariant I$^{\text{(BDI)}}(H)$ computes the parity of the integer-valued invariant predicted by Kitaev's table.

\section{Topological insulators: Mathematical setup and preliminaries}\label{4.1}

The origins of topological condensed matter can be traced back to the 1980s, with the discovery of the integer and fractional quantum Hall effects~\cite{vonKlitzing1986quantized, stormer1999fractional}. The quantization of the transverse conductivity in these phenomena defied conventional band theory, and its origin was soon understood to be rooted in topology, by exhibiting its relation to topological invariants such as the first Chern number~\cite{thouless1982quantized}. In the 2000s, the theoretical prediction and subsequent experimental realization of topological insulators brought topological phases into the spotlight once again~\cite{HasanKane2010, ando2013topological}. Topological insulators are materials that behave as insulators in their bulk but possess conducting states on their boundaries, which are protected by symmetries and topological invariants. Examples of $3$D topological insulators include Bi$_2$Se$_3$, Bi$_2$Te$_3$ and Sb$_2$Te$_3$: these materials exhibit spin-momentum locked surface states that are robust against non-magnetic disorder. Applications of topological insulators are being explored in fields such as spintronics, low-power electronics and quantum computing, particularly in the engineering of fault-tolerant qubits using Majorana modes at the edge of topological quantum chains~\cite{kitaev2001}. The unique transport properties and robustness to perturbations make these materials candidates for novel technological platforms. Future prospects include the integration of topological insulators in hybrid devices, the exploration of higher-order topological phases, and the manipulation of topological features using external parameters such as pressure or electromagnetic fields \cite{Sharma2022_TSCinterfaces,Xie2021_HOTIreview,Rajaji2022_PressureTIs,OkaKitamura2019_Floquet}.

To achieve a systematic understanding and control of topological phases, a rigorous classification is essential. Mathematically, non-interacting topological insulators can be modeled using Bloch bundles~\cite{panati2007triviality}: the presence of discrete symmetries, like time reversal, particle-hole, or chiral symmetry, modify the structure of the bundle and thus influence the classification, leading to a richer theory. The classification of Bloch bundles up to isomorphism or stable isomorphism naturally incorporates the use of $K$-theory~\cite{Kitaev_2009}, which allows one to assign well-defined invariants to these phases. We describe next and in the Appendices how this geometric picture arises from the quantum-mechanical description of crystalline solids. 

\subsection{Lattice Hamiltonians, their Bloch--Floquet fibers and their spectrum}

We model topological insulators mathematically within the one-body, non-interacting approximation of quantum mechanics, via tight-binding lattice Hamiltonians. Even though we will later specialize to quantum chains (so, 1D systems), we start from some general considerations which apply in any spatial dimension $d$.

We assume that the configuration space of the quantum system is modeled on a $d$-dimensional discrete lattice \( \mathbb{Z}^d \). To account for additional internal degrees of freedom, such as spin, orbital states or sublattice structure, we introduce \( N \) internal components per lattice site. The one-particle Hilbert space is thus taken to be
\begin{equation} \label{eq:tensor}
\mathcal{H} := \ell^2(\mathbb{Z}^d) \otimes \mathbb{C}^N,
\end{equation}
where \( \ell^2(\mathbb{Z}^d) \) encodes the spatial structure of the lattice and \( \mathbb{C}^N \) captures the internal degrees of freedom at each site.

We denote by \( U_\gamma : \mathcal{H} \to \mathcal{H} \), with \( \gamma \in \mathbb{Z}^d \), the unitary translation operators defined by
\begin{equation*}
(U_\gamma (\psi))_n := \psi_{n+\gamma}, \quad \text{for } \psi = (\psi_n) \in \ell^2(\mathbb{Z}^d) \otimes \mathbb{C}^N.
\end{equation*}
The lattice Hamiltonian $H$ is then required to be translation-invariant, i.e.\ such that $[H, U_\gamma]=0$ for all $\gamma \in \mathbb{Z}^d$. 
It is a standard result that a bounded self-adjoint operator \( H \) on~\( \mathcal{H} \) which is translation-invariant must act via convolution in the lattice variable~\cite{teschl2000jacobi}, that is, for any \( \psi = (\psi_n)_{n \in \mathbb{Z}^d} \in \mathcal{H} \), the action of \( H \) can be written as
\begin{equation} \label{eq:Hamiltonian}
(H (\psi))_n = \sum_{j \in \mathbb{Z}^d} A_j \psi_{n+j},
\end{equation}
for an appropriate family of matrices \( A_j \in \mathrm{M}_N(\mathbb{C}) \), indexed by \( j \in \mathbb{Z}^d \), with \( A_{-j} = A_j^\dagger \).

Physically, each matrix \( A_j \) can be interpreted as a matrix-valued hopping amplitude for a particle to move by a vector displacement \( j \) on the lattice.
Boundedness of \( H \) imposes decay conditions on the matrix norms \( \|A_j\| \), requiring them to vanish at infinity sufficiently fast. To simplify the setting, we shall restrict ourselves to the case of finite-range operators, where hopping occurs only within a bounded neighborhood. This is physically motivated for a number of reasons: first, because in real systems, such as crystalline solids, electrons move on an atomic lattice, and the probability of hopping from one site to another decays rapidly with distance due to the spatial localization of LCAO (Linear Combination of Atomic Orbitals) and the screening of interactions \cite{Ashcroft:102652,mahan2000,giuliani2005,altland2010}. Indeed, the hopping matrices \( A_j \) typically decay as
\begin{equation*}
\| A_j \| \sim \mathcal{O}(e^{-\xi |j|_{\ell^1}}),
\end{equation*}
with \( \xi > 0 \), so contributions from distant hopping decay exponentially and become negligible for practical purposes. Moreover, the rapid decay of long-range interactions is also due to the screening of Coulomb interactions by the surrounding electronic cloud. In a many-electron system, the effective interaction between charges is given by a screened Coulomb potential, i.e. a Yukawa-type potential,
\begin{equation*}
V(r) \propto \frac{e^{-\kappa r}}{r},
\end{equation*}
where \( \kappa^{-1} \) is the screening length, determined by the material's electronic properties; \( \kappa \) can be viewed as an effective mass parameter for the electric field in the medium. From a modeling perspective, this justifies including only neighbors at finite distance. Moreover, this approximation preserves the local symmetries of the lattice and significantly simplifies the mathematical analysis, without affecting the system's topological properties. Finally, in the thermodynamic and low energy limits, long range hopping contributions are negligible for many physically relevant observables \cite{bernevig2013,resta1994}. Therefore, we assume that there exists \( R \in \mathbb{N} \) such that
\begin{equation*}\label{consuA}
A_j = 0 \qquad \forall \  j \in \mathbb{Z}^d\setminus \{0\} \ : \  |j|_{\ell^1} > R.
\end{equation*}

Exploiting lattice translation invariance allows for a partial diagonalization of the Hamiltonian, and therefore a simplification of the corresponding spectral analysis. The transform that decomposes the Hamiltonian in fiber operators is the \emph{Bloch--Floquet transform} (essentially a vector-valued Fourier transform), defined explicitly by 
\begin{equation}\label{tra}
\mathcal{F}_d: \ell^2(\mathbb{Z}^d) \otimes \mathbb{C}^N \to L^2(\mathbb{T}^d; \mathbb{C}^N), \quad \mathcal{F}_d(e_{{n}} \otimes v) := \frac{e^{i n \cdot k}}{(\sqrt{2\pi})^{d}}  v,
\end{equation}
where $n \in \mathbb{Z}^d, v \in \mathbb{C}^N\,$, $\mathbb{T}^d$ stands for the (Brillouin) torus $\mathbb{T}^d := (S^1)^{\times d}$, and $\{(e_n)_{x} := \delta_{n,x}\}_{n \in \mathbb{Z}^d}$ denotes the canonical basis of $\ell^2(\mathbb{Z}^d)$. The above $\mathcal{F}_d$ is a bijective linear isometry because it sends an orthonormal basis to another orthonormal basis. Further properties of the transform are collected in Appendix~\ref{2.1}. For Hamiltonians of the form~\eqref{eq:Hamiltonian}, the Bloch--Floquet decomposition leads to the following result~\cite[Lemma~II.5]{Monaco_2023}.
\begin{Proposition}[Fiber Hamiltonians]\label{BFexpl}
    The operator $\mathcal{F}_dH\mathcal{F}_d^{-1}$ is a fibered operator,
    \[ \mathcal{F}_dH\mathcal{F}_d^{-1} = \int^{\oplus}_{\mathbb{T}^d} H_k \, dk\,, \]
    whose fiber Hamiltonians are matrices $H_k \in \mathrm{M}_N(\mathbb{C})$ given by $$H_k= \sum_{j \in J} e^{-i j \cdot k} A_j , \qquad k \in \mathbb{T}^d,$$ where $J:=\{j\in \mathbb{Z}^d \ : \ |j|_{\ell^1}\leq R\}$.
\end{Proposition}

By the general theory and the explicit representation in Proposition \ref{BFexpl}, the family of matrices \( \{H_k\} \) is $2\pi$-periodic in each component of $k$ and depends analytically on \( k \in \mathbb{T}^d \). In particular, the spectra of these matrices are uniformly bounded from below: for later reference, we denote by $\lambda_0$ a common lower bound for the spectra of the matrices $H_k$, i.e., 
\[ \lambda_0 < \inf_{k \in \mathbb{T}^d} \inf \sigma(H_k)\,. \]
This condition is essentially a physical one: the energy of the system cannot be unbounded from below since otherwise the condensed matter system will be unstable. This condition can be relaxed to consider Dirac-like models in which the infinity of negative energy states is interpreted as the Dirac sea. This interpretation finds its true fulfillment in Quantum Field Theory in which infinite negative energies are the prerogative of renormalization theory \cite{Peskin:1995ev, Weinberg:1995mt}. In the context of topological insulators, systems of Dirac type have been investigated e.g.\ in regards to their transport properties \cite{LuScSt13} or in relation to the relevance of the Zak phase as a topological probe \cite{Angelone2026}.

To model topological insulators, we formulate the following \emph{spectral gap hypothesis}: There exist \( \mu \in \mathbb{R} \) and \( g > 0 \) such that
\[ \inf_{k \in \mathbb{T}^d} \inf \sigma(H_k) < \mu - \frac{g}{2}, \quad \mathrm{and} \quad \left( \mu - \frac{g}{2}, \mu + \frac{g}{2} \right) \cap \sigma(H_k) = \emptyset \qquad \forall \ k \in \mathbb{T}^d\,. \]
Also this condition is a physical one since it tells us that the spectrum of the Hamiltonian is not connected but rather it is interspersed with spectral gaps: for example, a gap around the Fermi energy $\mu$ will separate the valence bands from the conduction bands.

Without loss of generality (see also the discussion below on quantum symmetries), we set $\mu = 0$. Thanks to the spectral gap hypothesis, the negative and positive eigenspaces of~\( H_k \), or rather the corresponding spectral eigenprojections, inherit the $2\pi$-periodicity in each component of $k$ and depend analytically on $k \in \mathbb{T}^d$~\cite{kato1976perturbation, Panati_2013}. To define the eigenprojections, we employ the Riesz formula: Let \( \alpha \subset \mathbb{C} \) be a simple, closed, positively oriented contour in the complex energy plane lying in the joint resolvent of the matrices $\{H_k\}_{k \in \mathbb{T}^d}$, and such that $\alpha \cap \mathbb{R} = \{ \lambda_0, \mu \} = \{\lambda_0, 0\}$. Set
\begin{equation} \label{eq:riesz}
P^{(-)}_k := \frac{i}{2\pi} \int_{\alpha}  (H_k-z1_{\mathbb{C}^N})^{-1}d{z}, \qquad k \in \mathbb{T}^d\,.
\end{equation}
Note that the integration in \eqref{eq:riesz} is performed in the complex energy plane, not in momentum space; moreover the matrix-valued integral is understood entry-wise. The operator \( P^{(-)}_k \) defined by the Riesz formula \eqref{eq:riesz} is an orthogonal projection that projects onto the direct sum of the eigenspaces associated with the negative eigenvalues of $H_k$, as $\alpha$ encloses only the negative part of the spectrum of $H_k$, where the resolvent $( H_k-z 1_{\mathbb{C}^N})^{-1}$ has its poles. By functional calculus, it follows that $P^{(-)}_k$ inherits the analyticity and periodicity properties of $H_k$ with respect to $k \in \mathbb{T}^d$~\cite{Panati_2013}. The complementary spectral eigenprojection onto the positive part of the spectrum is given by
\begin{equation}\label{P+}
    P^{(+)}_k := {1}_{\mathbb{C}^N} - P^{(-)}_k.
\end{equation}

As sketched in Appendix~\ref{2.1} and references therein, the data of the family of eigenprojections $\{P_k\}$, labelled by quasi-momentum $k$, allows to define a vector bundle over the Brillouin torus $\mathbb{T}^d$, called the Bloch bundle. It is this bundle that is responsible for the ``topological'' content of a topological insulator.

\subsection{Topological insulators and the ten-fold way}\label{3.3}

Several classification schemes have been proposed to distinguish topological phases of quantum matter. In a nutshell, two topological insulators are said to be in the same topological phase if their ground state projections can be continuously\footnote{The question about the topology which should be used on the space of allowed Hamiltonians to interpret a deformation as ``continuous'' is debated~\cite{thiang2015topological}, and may depend on the features one wants to include in the model, like disorder~\cite{shapiro2020topology} or locality~\cite{chung2025topological}. In the present context, norm-resolvent topology is enough.} deformed into one another without closing the spectral gap. Mathematically speaking, this corresponds to the existence of a continuous\footnote{In the standard norm topology since we are now talking about the spectral eigenprojections.} homotopy  between the respective families of spectral eigenprojections: according to the homotopy classification theorem~\cite{Husemoller1994}, this in turn implies that the associated Bloch bundles are isomorphic. However, a complete classification of this sort is achievable only in some low dimensional cases~\cite{Husemoller1994, panati2007triviality, monaco2023topology} and a $K$-theoretic approach gained prominence, in which only the $K^0$-class of the Bloch bundle (or equivalently, its stable equivalence class) is considered relevant for the identification of the quantum phase~\cite{Husemoller1994, kitaev2001}. 

The presence of further quantum symmetries may impose certain non-trivial constraints on the family of spectral eigenprojections, which in turn
can affect the topology of the corresponding Bloch bundle, either trivializing or enriching it. The type of discrete symmetries that are most relevant for topological insulators are time reversal symmetry $T$, particle-hole conjugation symmetry $C$ and chiral
symmetry $S$. These are symmetries in a somewhat generalized sense, in that they may commute or anticommute with the quantum Hamiltonian: see Appendix \ref{APPB} for details. As a result, topological insulators can be categorized into distinct symmetry classes, each determined by the presence or absence of certain discrete symmetries: these are often referred to as Altland--Zirnbauer--Cartan (AZC) classes~\cite{tab, heinzner2005symmetry}. Taking into account all possible combinations of these symmetry properties leads to a total of ten distinct symmetry classes. This classification scheme is known as the ``ten-fold way''~\cite{kitaev2001, Ryu_2010}, and forms the foundation for understanding topological insulators in condensed matter physics. 

The following Table~\ref{tabclass} reports this classification, obtained from the computation of appropriate (twisted equivariant) $K$-theory groups in each symmetry class and spatial dimension.
\begin{table}[h!]
\centering
\caption{Periodic table of topological insulators. The label for each row is a symmetry class following the AZC classification; the columns $T$, $C$, $S$ show the presence or not of the respective discrete symmetry as well as the value of its square (+1 if it is $+1_{\mathcal{H}}$, $-$1 if it is $-1_{\mathcal{H}}$); the last eight columns label the spatial dimensions. In view of Bott's periodicity, the table repeats periodically after $d=8$.}\label{tabclass}
\begin{tabular}{@{}l|ccc|ccccccccc@{}}
\toprule
\textbf{Symmetry Class} & $T$ & $C$ & $S$ & \textbf{1} & \textbf{2} & \textbf{3} & \textbf{4} & \textbf{5} & \textbf{6} & \textbf{7} & \textbf{8} \\
\midrule
A     & 0 & 0 & 0 & 0 & $\mathbb{Z}$ & 0 & $\mathbb{Z}$ & 0 & $\mathbb{Z}$ & 0 & $\mathbb{Z}$ \\
AIII  & 0 & 0 & 1 & $\mathbb{Z}$ & 0 & $\mathbb{Z}$ & 0 & $\mathbb{Z}$ & 0 & $\mathbb{Z}$ & 0 & \\
\midrule
AI    & 1 & 0 & 0 & 0 & 0 & 0 & $\mathbb{Z}$ & 0 & $\mathbb{Z}_2$ & $\mathbb{Z}_2$ & $\mathbb{Z}$ & \\
BDI   & 1 & 1 & 1 & $\mathbb{Z}$ & 0 & 0 & 0 & $\mathbb{Z}$ & 0 & $\mathbb{Z}_2$ & $\mathbb{Z}_2$ & \\
D     & 0 & 1 & 0 & $\mathbb{Z}_2$ & $\mathbb{Z}$ & 0 & 0 & 0 & $\mathbb{Z}$ & 0 & $\mathbb{Z}_2$ & \\
DIII  & -1 & 1 & 1 & $\mathbb{Z}_2$ & $\mathbb{Z}_2$ & $\mathbb{Z}$ & 0 & 0 & 0 & $\mathbb{Z}$ & 0 & \\
AII   & -1 & 0 & 0 & 0 & $\mathbb{Z}_2$ & $\mathbb{Z}_2$ & $\mathbb{Z}$ & 0 & 0 & 0 & $\mathbb{Z}$ & \\
CII   & -1 & -1 & 1 & $\mathbb{Z}$ & 0 & $\mathbb{Z}_2$ & $\mathbb{Z}_2$ & $\mathbb{Z}$ & 0 & 0 & 0 & \\
C     & 0 & -1 & 0 & 0 & $\mathbb{Z}$ & 0 & $\mathbb{Z}_2$ & $\mathbb{Z}_2$ & $\mathbb{Z}$ & 0 & 0 & \\
CI    & 1 & -1 & 1 & 0 & 0 & $\mathbb{Z}$ & 0 & $\mathbb{Z}_2$ & $\mathbb{Z}_2$ & $\mathbb{Z}$ & 0 & \\
\bottomrule
\end{tabular}
\end{table} 

\subsection{The action of the symmetry operators on fibres}

Discrete symmetries have an impact also on the spectral features of the quantum system and of its modelling Hamiltonian. In the notation of \eqref{eq:tensor}, we require, as is customary, that the discrete symmetry operators act as the identity on $\ell^2(\mathbb{Z}^d)$. We call $\mathfrak{T},\mathfrak{C},\mathfrak{S}$ the corresponding operators acting non-trivially on the $\mathbb{C}^N$-leg of the tensor product Hilbert space $\mathcal{H}$.

From the relations between the $T$, $S$ and $C$ operators and the Hamiltonian (see Appendix \ref{APPB}), using the Bloch--Floquet representation we can derive how their fiberwise counterparts act on the fibered operators and, thanks to the Riesz formula \eqref{eq:riesz}, on the spectral eigenprojections.
\begin{Proposition}[Discrete symmetries and fiber operators]\label{dsfo} Let $T,C,S$ be, respectively, the time reversal, particle-hole and chiral symmetry operator acting on the Hilbert space~$\mathcal{H}$
. Let $\mathfrak{T}$, $\mathfrak{C}$ and $\mathfrak{S}$, respectively, be their $k$-independent fiber operators with respect to the Bloch--Floquet transform
. Let $H_k$ the fiber Hamiltonian and let $P_k$ be the fiber of the spectral eigenprojection defined by the Riesz formula that projects onto all the negative energy bands. Then
    \begin{enumerate}
    \item 
    \(
    H_{-k} \mathfrak{T} = \mathfrak{T}H_{k}, \quad \mathrm{and}  \quad  P_k \mathfrak{T} = \mathfrak{T}P_{-k}; 
    \)
    
    \item 
    \(
     H_{-k} \mathfrak{C} = -\mathfrak{C}H_{k}, \quad \mathrm{and}  \quad  P_k \mathfrak{C} = \mathfrak{C}(1_{\mathbb{C}^N} - P_{-k});
    \)
    
    \item 
    \(
     H_{k} \mathfrak{S} = -\mathfrak{S}H_k, \quad \mathrm{and}  \quad  P_k \mathfrak{S} = \mathfrak{S}(1_{\mathbb{C}^N} - P_k).
    \)
\end{enumerate}
\end{Proposition}
\begin{proof}
    The relations in their first form are inherited by the fiber operators directly from the commutativity or anti-commutativity properties of the operators $T,C,S$ with the Hamiltonian. Similarly, concerning the relations with the fiber operators of the spectral eigenprojections, these follow directly from the former using the Riesz formula. 
\end{proof}
\begin{table}[h!]
\centering
\caption{Summary of the symmetry operators $\mathfrak{T}$, $\mathfrak{C}$, $\mathfrak{S}$ and their possible combinations. We report their actions on quasi-momentum $k$ and on the fiber Hamiltonian $H_k$. The last column shows the resulting spectral symmetry.}\label{tabriass}
\renewcommand{\arraystretch}{1.4}
\begin{tabular}{l|l|l|l}
\toprule
\textbf{Operator} & \textbf{Action on} $k$  & \textbf{Action on} $H_k$ & \textbf{Spectral Sym.} \\
\midrule
$\mathfrak{T}$ & $k \mapsto -k$ &  $H_{-k} = \mathfrak{T} H_k \mathfrak{T}^{-1}$ & $E \leftrightarrow E$ \\
$\mathfrak{C}$ & $k \mapsto -k$ & $H_{-k} = -\mathfrak{C} H_k \mathfrak{C}^{-1}$ & $E \leftrightarrow -E$ \\
$\mathfrak{S}$ & $k \mapsto k$   &  $H_k = -\mathfrak{S} H_k \mathfrak{S}^{-1}$ & $E \leftrightarrow -E$ \\
$\mathfrak{T} \mathfrak{C}$ & $k \mapsto k$   &  $H_k=-(\mathfrak{T} \mathfrak{C}) H_k (\mathfrak{T} \mathfrak{C})^{-1}$ & $E \leftrightarrow -E$ \\
$\mathfrak{T} \mathfrak{S}$ & $k \mapsto -k$ &  $H_{-k}=-(\mathfrak{T} \mathfrak{S}) H_k (\mathfrak{T} \mathfrak{S})^{-1}$ & $E \leftrightarrow -E$ \\
$\mathfrak{C} \mathfrak{S}$ & $k \mapsto -k$ &  $H_{-k}=(\mathfrak{C} \mathfrak{S}) H_k (\mathfrak{C} \mathfrak{S})^{-1} $ & $E \leftrightarrow E$ \\
$\mathfrak{T} \mathfrak{C} \mathfrak{S}$ & $k \mapsto k$ &  $H_k=(\mathfrak{T} \mathfrak{C} \mathfrak{S}) H_k (\mathfrak{T} \mathfrak{C} \mathfrak{S})^{-1}$ & $E \leftrightarrow E$ \\
\bottomrule
\end{tabular}
\end{table}
The above Table \ref{tabriass} summarizes the action of the discrete symmetry operators and their possible combinations, together with the spectral consequences that their presence entails. The table shows in particular that in symmetry classes BDI, DIII, CII and CI there are always symmetries with a non-trivial action on the spectrum: particle-hole symmetry, chiral symmetry and their combinations with time reversal symmetry map $E$ to $-E$. It appears that symmetries and their combinations can be classified according to two features: whether they are unitary or antiunitary (and therefore whether they keep quasi-momentum~$k$ fixed, or they map it to $-k$), and whether they commute or anticommute with the Hamiltonian (and therefore whether they keep the energy $E$ fixed, or they map it to $-E$). Accordingly, we denote the general symmetry combination as $\mathfrak{O}^{(\varepsilon_k \, \varepsilon_E)}$, where $\varepsilon_k, \varepsilon_E \in \{-, +\}$ are the signs corresponding to $k \mapsto \varepsilon_k \, k$ and $E \leftrightarrow \varepsilon_E E$; with this notation, symmetries fall into four classes:
\[ \begin{split} \mathfrak{O}^{(++)} \in \{1_{\mathbb{C}^N},\mathfrak{T}\mathfrak{C}\mathfrak{S}\}\,, \qquad &
\mathfrak{O}^{(+-)} \in \{\mathfrak{S},\mathfrak{T}\mathfrak{C}\}\,, \qquad\\
\mathfrak{O}^{(-+)} \in \{\mathfrak{T},\mathfrak{C}\mathfrak{S}\}\,, \qquad
& \mathfrak{O}^{(--)} \in \{\mathfrak{C},\mathfrak{T}\mathfrak{S}\}\,. \end{split}\]

Moreover, the action of the discrete operators on the spectral eigenprojections now reads as follows.
\begin{Proposition}[Discrete symmetries and spectral eigenprojections]\label{DSSP}
Let \( P^{(-)}_k \) and $P^{(+)}_k$ be as in equations \eqref{eq:riesz} and \eqref{P+}. Then
\[ \mathfrak{O}^{(\varepsilon_k \, \varepsilon_E)} P_k^{(\pm)} = P_{\varepsilon_k \, k}^{(\pm \varepsilon_E)}\mathfrak{O}^{(\varepsilon_k \, \varepsilon_E)}\,, \quad \varepsilon_k \,, \varepsilon_E \in \{-, +\}\,. \]
\end{Proposition}

According to the theory of fibered operators outlined in Appendix \ref{2.1}, the spectrum of the original Hamiltonian $H$ is given by the union of the spectra of the fiber Hamiltonians \( H_k \) with \( k \in \mathbb{T}^d \). Therefore, in presence of the spectral symmetry $E \leftrightarrow -E$, the spectrum \( \sigma(H) \) must be symmetric with respect to \( E = 0 \). This is why we impose the relevant spectral gap to be around zero energy.

We observe that the quasi-momentum \( k = 0 \) is fixed under the involution \( \Phi: k \mapsto -k \). By the spectral theorem, \( \mathbb{C}^N \) admits an orthonormal basis \( v_1, \dots, v_N \) consisting of eigenvectors of the self-adjoint matrix \( H_0 \). Assuming a spectral gap around $\mu=0$, we may arrange that the first \( m \) vectors \( \{v_1, \dots, v_m\} \) correspond to negative eigenvalues, while the remaining vectors \( \{v_{m+1}, \dots, v_N\} \) correspond to positive eigenvalues. The presence of at least one symmetry operator $\mathfrak{O}^{(\varepsilon_k \, -)}$ inducing a spectral symmetry $E \leftrightarrow -E$ imposes the existence of a bijection between the eigenspaces of \( H_0 \) associated with negative and positive eigenvalues. As a consequence, it must be that \( m = N - m \), and hence the fiber dimension \( N \) must be even, $N=2m$.

The general theory presented in Appendix~\ref{APPB} allows for the presence of independent symmetries $\mathfrak{T}$, $\mathfrak{C}$ and $\mathfrak{S}$; however, the product of two of them is of the third type. Therefore, a system could in principle be symmetric under three such symmetries, but the combination $\mathfrak{T}\mathfrak{C}$ (or $i \, \mathfrak{T}\mathfrak{C}$) and the operator $\mathfrak{S}$ would both behave as chiral symmetries. In Kitaev's periodic table (see Table~\ref{tabclass}), it is instead assumed that the three symmetries are not independent, namely that the third one is determined by the product of the other two (up to multiplication times $\pm 1$ or $\pm i$, depending on whether the two symmetry operators commute or anticommute among themselves). We will follow the same lines: thus, for example, a symmetry operator of type $\mathfrak{O}^{(++)}$ is required to be proportional to the identity, and thus to describe the situation in which there is no symmetry at all (as in class A).

\section{Symmetric Bloch bases}\label{BB}

\subsection{General definition}
Later we will be interested in quantifying the topological content of the spectral eigenprojections of the fiber operators $H_k$ into a numerical invariant, and in particular in extracting this topological information from the geometric Zak phase (see Section~\ref{4.3}). To do so we start by identifying bases which behave appropriately under the discrete symmetries of the system.
\begin{Definition}[Bloch basis]
    Given a set of rank-$m$ projections \( \{P_k\} \subset \mathrm{M}_N(\mathbb{C}) \), which is $2\pi$-periodic in each component of $k$ and at least \( C^1 \) in \( k \), we define \( \{v_j(k)\}_{j=1}^N \subset \mathbb{C}^N \) to be a Bloch basis associated to \( \{P_k\} \) if it satisfies the following properties:
\begin{enumerate}
    \item each vector \( v_j(k) \) is \( 2\pi \)-periodic and as regular as \( P_k \) as a function of \( k \in \mathbb{T}^d \);
    \item for every \( k \in \mathbb{T}^d \), the collection \( \{v_j(k)\}_{j=1}^N \) forms an orthonormal basis of \( \mathbb{C}^N \);
    \item for every \( k \in \mathbb{T}^d \), the vectors \( \{v_j(k)\}_{j=1}^m \) span \( \operatorname{Range}(P_k) \) and therefore the vectors \( \{v_j(k)\}_{j=m+1}^N \) span \( \mathrm{Ker}(P_k) = \operatorname{Range}(1_{\mathbb{C}^N} - P_k) \).
\end{enumerate}
\end{Definition}
Since the systems under consideration enjoy some discrete symmetries, we want a Bloch basis to have well-defined transformation properties under the symmetry operators. The transformation properties are inferred from Table \ref{tabriass}. Therefore we give the following Definition. 
\begin{Definition}[Symmetric Bloch basis]\label{SBB}
    Let \( P^{(-)}_k \) and $P^{(+)}_k$ be as in \eqref{eq:riesz} and \eqref{P+}, and let $m$ be the rank of the projection $P^{(-)}_{k}$. Then, a Bloch basis \( \{v_j(k)\}_{j=1}^N \) associated to \( \left\{P^{(-)}_k\right\} \) is said to be a symmetric Bloch basis if the following are satisfied:
\begin{enumerate}
    \item let $\mathfrak{O}^{(++)} \in \{1_{\mathbb{C}^N},\mathfrak{T}\mathfrak{C}\mathfrak{S}\}$; then 
    \begin{equation*}
        \mathfrak{O}^{(++)}\, v_j(k) = v_{j}(k), \qquad \forall \ j \in \{1,...,N\}
    \end{equation*}
    \item let $\mathfrak{O}^{(+-)} \in \{\mathfrak{S},\mathfrak{T}\mathfrak{C}\}$; then $$\mathfrak{O}^{(+-)}\, v_j(k) = v_{N - j + 1}(k), \qquad \forall \ j \in \{1,...,m\};$$  
    \item let $\mathfrak{O}^{(-+)} \in \{\mathfrak{T},\mathfrak{C}\mathfrak{S}\}$; then
    \begin{equation*}
    \begin{aligned}
        &\mathfrak{O}^{(-+)}\, v_j(k) = v_{j}(-k), \qquad \forall \ j \in \{1,...,N\} \quad \mathrm{if} \ \left(\mathfrak{O}^{(-+)}\right)^2=+1_{\mathbb{C}^{N}};\\
        &
        \mathfrak{O}^{(-+)}\, v_j(k) =\sum_{\ell=1}^{2m} v_{\ell}(-k) \varepsilon_{j\ell}, \qquad \forall \ j \in \{1,...,2m\}  \quad \mathrm{if} \ \left(\mathfrak{O}^{(-+)}\right)^2=-1_{\mathbb{C}^{2m}};
    \end{aligned} 
    \end{equation*}
    \item let $\mathfrak{O}^{(--)} \in \{\mathfrak{C},\mathfrak{T}\mathfrak{S}\}$; then $$\mathfrak{O}^{(--)}\, v_j(k) = v_{N - j + 1}(-k), \qquad \forall \ j \in \{1,...,m\};$$
\end{enumerate}
where $\varepsilon = \begin{bmatrix}\varepsilon_{j\ell}\end{bmatrix}_{1 \le j,\ell \le 2m}$ is a reshuffling matrix which is unitary and anti-symmetric. 
\end{Definition}
The necessity of reshuffling matrices is evident if one looks at the invariant subspace at $k = 0$. In fact, it is easy to check that, for all $v(0) \in \operatorname{Range}(P^{(-)}_0)$, the vector $\mathfrak{O}^{(-+)}v(0)$ is orthogonal to $v(0)$ if the operator $\mathfrak{O}^{(-+)}$ squares to $-1_{\mathbb{C}^{2m}}$ (compare the comments under Definition~\ref{quaternionic} below). In particular this would imply the vanishing of the vectors and prevents from setting the definition used for operators that square to $+1_{\mathbb{C}^{2m}}$ as showed e.g.\ in \cite{graf2013bulk, Fiorenza_2016, Cornean_2017, Peluso2026}. 

In the same vein, we observe that the symmetry conditions presented in the above Definition may always be imposed at the single point $k=0$ (see the comments at the end of last Section, and compare~\cite{Monaco_2023} for chiral and  particle-hole symmetric classes, and~\cite{fiorenza2016construction, Fiorenza_2016} for time-reversal symmetric classes): for example, the identity $\mathfrak{O}^{(+-)} v_j(0) = v_{N-j+1}(0)$, $j \in \{1,\ldots,m\}$, may be seen as a definition of the vectors on the right-hand side, once the vectors $\{v_1(0),\ldots,v_m(0)\}$ spanning $\operatorname{Range}(P^{(-)}_0)$ orthonormally are given.

\subsection{Bloch bases in 1D through parallel transport}

Henceforth, we focus on the 1D case, so $d=1$. We ask whether a symmetric Bloch basis for a 1D family of projections exists. The question is non-trivial, as in the geometric picture having a (symmetric) Bloch basis is tantamount to having a trivializing frame (which is compatible with the symmetries) for the corresponding Bloch bundle. We will show that this is indeed always the case, using a tool also borrowed from the geometric viewpoint: parallel transport.
\begin{Theorem}[Parallel transport]\label{TP}
Let \( \{P_k\} \in \mathrm{M}_N(\mathbb{C}) \) be a family of projections which is at least \( C^1 \) in \( k \). Then there exists a unique family of operators \( \{\mathcal{T}_k\} \in \mathrm{M}_N(\mathbb{C}) \) that solves the Cauchy problem
\begin{equation*} \label{eq:Cauchy}
\begin{cases}
\partial_k\mathcal{T}_k = [\partial_kP_k, P_k] \mathcal{T}_k, \\
\mathcal{T}_0 = 1_{\mathbb{C}^N}.
\end{cases}
\end{equation*}
The family of operators \( \{\mathcal{T}_k\} \) is at least as regular as \( \{P_k\} \) as a function of \( k \). Moreover, \( \mathcal{T}_k \) satisfies the following properties:
\begin{enumerate}
    \item Unitarity. Each \( \mathcal{T}_k \) is unitary, i.e. \( (\mathcal{T}_k)^* = (\mathcal{T}_k)^{-1} \).
    \item Intertwining property. We have the relation
    \begin{equation*} \label{eq:intertwining}
    P_k \mathcal{T}_k = \mathcal{T}_k P_0, \qquad \forall \ k \in \mathbb{R}.
    \end{equation*}
    \item Telescopic property. If \( P_k \) is \( 2\pi \)-periodic in \( k \), then
    \begin{equation*} \label{eq:telescopic}
    \mathcal{T}_{k + 2\pi n} = \mathcal{T}_k (\mathcal{T}_{2\pi})^n, \qquad \forall \ k \in \mathbb{R}, \, n \in \mathbb{Z}.
    \end{equation*}
\end{enumerate}
We will call the family of unitaries \( \mathcal{T}_k \) the parallel transport operators associated with \( \{P_k\} \).
\end{Theorem}
\begin{proof}
The proof can be found in \cite[Theorem III.3]{Monaco_2023}.
\end{proof}
We note that the uniqueness of the solution of the Cauchy problem of Theorem \ref{TP} implies that the parallel transport operators associated with the families \( P^{(-)}_k \) and \( P^{(+)}_k = 1_{\mathbb{C}^N} - P^{(-)}_k \) coincide, since they have the same generator of the linear differential equation:
\begin{equation}\label{eqgen}
    \left[ \partial_k P^{(+)}_k, P^{(+)}_k \right] = \left[ \partial_k (1_{\mathbb{C}^N} - P^{(-)}_k), 1_{\mathbb{C}^N} - P^{(-)}_k \right] = \left[ \partial_k P^{(-)}_k, P^{(-)}_k \right].
\end{equation}
Using the same uniqueness property it is possible to show the following result.
\begin{Proposition}[Discrete symmetries and parallel transport]\label{propcomm}
    Let $\{P_k^{(-)}\}$ and $\{P_k^{(+)}\}$ be as in \eqref{eq:riesz} and \eqref{P+}. Let $\{\mathcal{T}_k\}$ be the family of parallel transport unitaries associated with $\{P_k^{(-)}\}$. Then:
\begin{enumerate}
    \item for $\mathfrak{O}^{(+)} \in \left\{\mathfrak{O}^{(++)},\mathfrak{O}^{(+-)}\right\}$, i.e. $\mathfrak{O}^{(+)} \in \{1_{\mathbb{C}^N},\mathfrak{S},\mathfrak{T}\mathfrak{C},\mathfrak{T}\mathfrak{C}\mathfrak{S}\}$, one has 
    $$\mathfrak{O}^{(+)} \mathcal{T}_k = \mathcal{T}_k\mathfrak{O}^{(+)};$$
    
    \item for $\mathfrak{O}^{(-)} \in \left\{\mathfrak{O}^{(-+)},\mathfrak{O}^{(--)}\right\}$, i.e. $\mathfrak{O}^{(-)} \in \{\mathfrak{T},\mathfrak{C},\mathfrak{T}\mathfrak{S},\mathfrak{C}\mathfrak{S}\}$, one has
    $$\mathfrak{O}^{(-)} \mathcal{T}_k = \mathcal{T}_{-k}\mathfrak{O}^{(-)}.$$
\end{enumerate}
\end{Proposition}
\begin{proof}
    We take advantage once again of the uniqueness of the solution to the Cauchy problem. For the first case, let us define $$U_k:=\left(\mathfrak{O}^{(+)} \right)^{-1}\mathcal{T}_k\mathfrak{O}^{(+)};$$  
differentiating, using Proposition \ref{DSSP} and \eqref{eqgen}, we get
\begin{equation*}
\begin{aligned}
    \partial_kU_k&=\left(\mathfrak{O}^{(+)} \right)^{-1}(\partial_k\mathcal{T}_k)\mathfrak{O}^{(+)}=\left(\mathfrak{O}^{(+)} \right)^{-1} \left[ \partial_k P^{(-)}_k, P^{(-)}_k \right] \mathcal{T}_k\mathfrak{O}^{(+)}=\\
    &=\left[\left(\mathfrak{O}^{(+)} \right)^{-1}\partial_k P_k^{(-)}\mathfrak{O}^{(+)},\left(\mathfrak{O}^{(+)} \right)^{-1}P_k^{(-)}\mathfrak{O}^{(+)}\right] \left(\mathfrak{O}^{(+)} \right)^{-1}\mathcal{T}_k\mathfrak{O}^{(+)}=\\
    &=\left[\partial_kP_k^{(+/-)},P_k^{(+/-)}\right]U_k;
\end{aligned}
\end{equation*} 
moreover, we have $U_0=\left(\mathfrak{O}^{(+)} \right)^{-1}\mathcal{T}_0\mathfrak{O}^{(+)}=1_{\mathbb{C}^N}.$ This means that $U_k$ solves the same Cauchy problem of $\mathcal{T}_k$, hence
$$\left(\mathfrak{O}^{(+)} \right)^{-1}\mathcal{T}_k\mathfrak{O}^{(+)}=\mathcal{T}_k.$$
Similarly, let us define $$U_k:=\left(\mathfrak{O}^{(-)} \right)^{-1}\mathcal{T}_{-k}\mathfrak{O}^{(-)};$$ we have
\begin{equation*}
\begin{aligned}
    \partial_kU_k&=\left(\mathfrak{O}^{(-)} \right)^{-1}(\partial_k\mathcal{T}_{-k})\mathfrak{O}^{(-)}=\left(\mathfrak{O}^{(-)} \right)^{-1} \left[ \partial_k P^{(-)}_{-k}, P^{(-)}_{-k} \right] \mathcal{T}_{-k}\mathfrak{O}^{(-)}=\\
    &=\left[\left(\mathfrak{O}^{(-)} \right)^{-1}\partial_k P_{-k}^{(-)}\mathfrak{O}^{(-)},\left(\mathfrak{O}^{(-)} \right)^{-1}P_{-k}^{(-)}\mathfrak{O}^{(-)}\right] \left(\mathfrak{O}^{(-)} \right)^{-1}\mathcal{T}_{-k}\mathfrak{O}^{(-)}=\\
    &=\left[\partial_kP_k^{(+/-)},P_k^{(+/-)}\right]U_k;
\end{aligned}
\end{equation*} 
moreover, we have $U_0=\left(\mathfrak{O}^{(+)} \right)^{-1}\mathcal{T}_0\mathfrak{O}^{(+)}=1_{\mathbb{C}^N}.$ Therefore 
\[\left(\mathfrak{O}^{(-)} \right)^{-1}\mathcal{T}_{-k}\mathfrak{O}^{(-)}=\mathcal{T}_k.\qedhere\]
\end{proof}
Our goal is now to exhibit a symmetric Bloch basis. The following result shows the explicit construction.
\begin{Theorem}[Existence of a symmetric Bloch basis]\label{exbas}
    Let \( P^{(-)}_k \in \mathrm{M}_N(\mathbb{C}) \) be the rank-\( m \) eigenprojection defined by the Riesz formula in \eqref{eq:riesz}. Let also   \( \{v_1(0), ..., v_N(0)\} \in \mathbb{C}^N \) be an orthonormal basis of eigenvectors of \( H_0 \) such that  
\( \{v_1(0), ..., v_m(0)\} \) spans \( \operatorname{Range}\left(P^{(-)}_0\right) \) while the others span $\mathrm{Ker}\left(P^{(-)}_0\right) = \operatorname{Range}\left(1_{\mathbb{C}^N} - P^{(+)}_0\right)$.  
Finally, denote by \( \{\mathcal{T}_k\} \) the family of parallel transport unitaries associated with \( \left\{P^{(-)}_k\right\} \).  
Then the family of vectors
\begin{equation} \label{eq:BlochBasis}
v_j(k) := \mathcal{T}_k e^{-i k X / 2\pi} v_j(0), \qquad j \in \{1, \ldots, N\},    
\end{equation}
where 
$$X:=-i \ln(\mathcal{T}_{2\pi}),$$ is a symmetric Bloch basis.
\end{Theorem}
\begin{proof}
The proof follows that of \cite[Theorem III.6]{Monaco_2023}, which focuses only on chiral and particle-hole symmetric projections. Let us first discuss the definition of $X$. We focus on the unitary matrix \( \mathcal{T}_{2\pi} \in \mathrm{M}_N(\mathbb{C}) \), called the holonomy unitary. By the spectral theorem it can be diagonalized by a unitary transformation \( V \),
\[
\mathcal{T}_{2\pi} = V^{-1}
\begin{bmatrix}
e^{i\varphi_1} & & 0 \\
& \ddots & \\
0 & & e^{i\varphi_N}
\end{bmatrix}
V, \qquad \varphi_j \in [0, 2\pi), \ \forall \ j \in \{1, \dots, N\}.
\]
With this choice of the phases for the eigenvalues of \( \mathcal{T}_{2\pi} \), we define \( X \in \mathrm{M}_N(\mathbb{C}) \) to be the self-adjoint matrix
\[
X := -i \ln (\mathcal{T}_{2\pi}) = V^{-1}
\begin{bmatrix}
\varphi_1 & & 0 \\
& \ddots & \\
0 & & \varphi_N
\end{bmatrix}
V,
\]
 so that \( \mathcal{T}_{2\pi} = e^{iX} \).

The spanning property, regularity, and periodicity of the vectors in \eqref{eq:BlochBasis} are established as in \cite[Theorem III.6]{Monaco_2023}. We therefore focus on the symmetry properties: as stated after Definition~\ref{SBB}, they can be assumed to hold at $k=0$. First, notice that the commutation relations in Proposition~\ref{propcomm} can be inherited by the logarithm of the holonomy unitary, namely there exists a choice of $X$ as above such that
\[ \mathfrak{O}^{(\pm)} X = \pm X \mathfrak{O}^{(\pm)}. \]
Next, we consider the four cases $\mathfrak{O}^{(++)},\mathfrak{O}^{(+-)},\mathfrak{O}^{(-+)},\mathfrak{O}^{(--)}$ separately. Let us first consider the case of $\mathfrak{O}^{(++)}$. Since $\mathfrak{O}^{(++)}$ commutes with $\mathcal{T}_k$, in particular it commutes with $\mathcal{T}_{2\pi}$ and, by functional calculus, with $e^{-ikX/2\pi}$. Therefore
\begin{equation*}
\begin{aligned}
\mathfrak{O}^{(++)}v_j(k)&=\mathfrak{O}^{(++)}\mathcal{T}_ke^{-ikX/2\pi}v_j(0)=\mathcal{T}_ke^{-ikX/2\pi}\mathfrak{O}^{(++)}v_j(0)=\\
&=\mathcal{T}_ke^{-ikX/2\pi}v_{j}(0)
=v_{j}(k).
\end{aligned}  
\end{equation*}
The case $\mathfrak{O}^{(+-)}$ is similar since it also commutes with $\mathcal{T}_k$ but we have
\begin{equation*}
\begin{aligned}
\mathfrak{O}^{(+-)}v_j(k)&=\mathfrak{O}^{(+-)}\mathcal{T}_ke^{-ikX/2\pi}v_j(0)=\mathcal{T}_ke^{-ikX/2\pi}\mathfrak{O}^{(+-)}v_j(0)
=\\ &=\mathcal{T}_ke^{-ikX/2\pi}v_{N-j+1}(0)=v_{N-j+1}(k).
\end{aligned}  
\end{equation*}
In the case of $\mathfrak{O}^{(-+)}$, it commutes with $\mathcal{T}_k$ up to changing sign to $k$, hence, by functional calculus and the previous choice of $X$, the same holds with $e^{-ikX/2\pi}$. Therefore if $\left(\mathfrak{O}^{(-+)}\right)^2=+1_{\mathbb{C}^{2m}}$ we have
\begin{equation*}
\begin{aligned}
\mathfrak{O}^{(-+)}v_j(k)&=\mathfrak{O}^{(-+)}\mathcal{T}_ke^{-ikX/2\pi}v_j(0)=\mathcal{T}_{-k}e^{ikX/2\pi}\mathfrak{O}^{(-+)}v_j(0)=\\
&=\mathcal{T}_{-k}e^{ikX/2\pi}v_{j}(0)=v_{j}(-k);
\end{aligned}  
\end{equation*}
while if $\left(\mathfrak{O}^{(-+)}\right)^2=-1_{\mathbb{C}^{2m}}$ we have
\begin{equation*}
\begin{aligned}
\mathfrak{O}^{(-+)}v_j(k)&=\mathfrak{O}^{(-+)}\mathcal{T}_ke^{-ikX/2\pi}v_j(0)=\mathcal{T}_{-k}e^{ikX/2\pi}\mathfrak{O}^{(-+)}v_j(0)=\\
&=\mathcal{T}_{-k}e^{ikX/2\pi}\sum_{\ell=1}^Nv_{\ell}(0)\varepsilon_{j\ell}=\sum_{\ell=1}^N\mathcal{T}_{-k}e^{ikX/2\pi}v_{\ell}(0)\varepsilon_{j\ell}=\\
&=\sum_{\ell=1}^Nv_{\ell}(-k)\varepsilon_{j\ell}.
\end{aligned}  
\end{equation*}
Finally, in a similar way we have
\begin{equation*}
\begin{aligned}
\mathfrak{O}^{(--)}v_j(k)&=\mathfrak{O}^{(--)}\mathcal{T}_ke^{-ikX/2\pi}v_j(0)=\mathcal{T}_{-k}e^{ikX/2\pi}\mathfrak{O}^{(--)}v_j(0)=\\ &=\mathcal{T}_{-k}e^{ikX/2\pi}v_{N-j+1}(0)=v_{N-j+1}(-k). \qedhere
\end{aligned}  
\end{equation*}
\end{proof}
\section{The invariant}\label{4.3}
The topological invariant associated with the spectral eigenprojections $P_k^{(-)}$ of the fiber Hamiltonians $H_k$ will be extracted from any symmetric Bloch basis, whose existence is guaranteed by Theorem \ref{exbas}, starting form the abelian
Zak phase \cite{zak1989berry} of the Bloch basis.
\subsection{The Zak phase and its properties} \label{def:ZakPhase}
We start by recalling the definition of the abelian Zak phase.
\begin{Definition}[Zak phase]
    Let \( P_k \in \mathrm{M}_N(\mathbb{C}) \) be a \( 2\pi \)-periodic family of rank-\( m \) projections which is at least \( C^1 \) in \( k \), and let \( \{v_j(k)\}_{j=1}^N \) be a Bloch basis associated with \( P_k \). The abelian Zak phase is defined as
\[
\mathcal{Z}^{(N)}_{\{v_j(k)\}}:=\frac{1}{2\pi i} \int_{S^1} \sum_{j=1}^{N} \left\langle v_j(k), \partial_k v_j(k) \right\rangle dk.
\]
\end{Definition}
The Zak phases of the Bloch basis constructed via parallel transport in Theorem \ref{exbas} is always an integer: this is stated in 
\begin{Theorem}
\label{BFBB}
    Let $$v_j(k) := \mathcal{T}_k e^{-i k X / 2\pi} v_j(0), \qquad j \in \{1, \ldots, N\}$$ be the Bloch basis from Theorem \ref{exbas} with $$X:=-i \ln \mathcal{T}_{2\pi}.$$ Then
\[
\mathcal{Z}^{(N)}_{\{v_j(k)\}}=\frac{1}{2\pi i} \int_{S^1} \sum_{j=1}^{N} \langle v_{j}(k), \partial_k v_{j}(k) \rangle dk = -\frac{1}{2\pi} \operatorname{tr} (X) \in \mathbb{Z}.
\]
\end{Theorem}
\begin{proof}
The proof can be found in \cite[Theorem IV.3]{Monaco_2023}.
\end{proof}

We are now interested in considering a change of Bloch basis. We first give the following Definition.
\begin{Definition}[Bloch gauge]
    Let $\{v_j(k)\}_{j=1}^N$ be a Bloch basis and let $\{u_j(k)\}_{j=1}^N$ be another Bloch basis. The change of basis matrix $B_k$ such that $u_j(k)=B_kv_j(k)$ for all $i \in \{1,. . . , N\}$ is called Bloch gauge.
\end{Definition}
By definition, a Bloch basis $\{u_j(k)\}_{j=1}^N$ associated to $P^{(-)}_k$ has the first $m$ vectors in $\operatorname{Range}\left(P^{(-)}_k\right)$ and the last $N - m$ in $\operatorname{Range} \left(P^{(+)}_k\right)$. Therefore, the Bloch gauge $B_k$ which maps the Bloch basis $\{v_j(k)\}_{j=1}^N$ to $\{u_j(k)\}_{j=1}^N$ has a block-diagonal form in the decomposition \[\mathbb{C}^N = \operatorname{Range}\left(P^{(-)}_k\right) \oplus \operatorname{Range}\left(P^{(+)}_k\right),\] i.e.
\[
B_k =
\begin{bmatrix}
B^{(-)}_k & 0 \\
0 & B^{(+)}_k
\end{bmatrix};
\]
where $B^{(\pm)}_k = P^{(\pm)}_k B_k P^{(\pm)}_k$, seen as a linear operator on $\operatorname{Range} \left(P^{(\pm)}_k\right)$. In particular,
\[
u_j(k) =
\begin{cases}
B^{(-)}_k v_j(k)  & \qquad \text{if } \quad j \in \{1, \ldots, m\}, \\
B^{(+)}_k v_j(k) & \qquad \text{if } \quad j \in \{m+1, \ldots, N\}.
\end{cases}
\]
The next result shows that, when a change of Bloch basis is performed, the Zak phase changes by an integer. 
\begin{Theorem}[Bloch basis change and Zak phase]\label{bbcbf}
Let $\{v_j(k)\}_{j=1}^N$ be a Bloch basis and let $\{u_j(k)\}_{j=1}^N$ be another Bloch basis. Let $B_k$ the corresponding Bloch gauge. Then $B_k$ is unitary, $2\pi$-periodic and as regular in $k$ as the Bloch bases. Moreover,
\begin{equation} \label{eqn:wnB}
\mathcal{Z}^{(N)}_{\{u_j(k)\}}-\mathcal{Z}^{(N)}_{\{v_j(k)\}}
= \mathrm{wn}([B]) \in \mathbb{Z},
\end{equation}
where $\mathrm{wn}([B])$ denotes the winding number of the map $$B : S^1 \to U(N), \quad k \mapsto B_k,$$
compare Appendix~\ref{appendixC}.
\end{Theorem}
\begin{proof}
The proof can be found in \cite[Theorem IV.5]{Monaco_2023}.
\end{proof}
In combination with the previous result, we can conclude that Zak phases are always integer-valued. 
\subsection{The $\mathbb{Z}_2$ invariant}
The point to be addressed now is by how much the Zak phase of a Bloch basis can vary under a change of Bloch gauge, that is, what values the term $\mathrm{wn}([B])$ in \eqref{eqn:wnB} can take. In case only a chiral or particle-hole symmetry is present, then $\mathrm{wn}([B])\in 2 \mathbb{Z}$ is necessarily even, and so taking the mod-2 remainder of the Zak phase we get a $\mathbb{Z}_2$ gauge-invariant quantity, that can be further shown to be a topological invariant: this is the main result of~\cite{Monaco_2023}. Our contribution generalizes the analysis to all possible discrete symmetries and their combinations.

\begin{Theorem}[Gauge transformations and Zak phase]\label{2Zwi}
    Consider $\{v_j(k)\}_{j=1}^N$ a symmetric Bloch basis and let $\{u_j(k)\}_{j=1}^N$ be another symmetric Bloch basis. Let $B_k$ be the corresponding Bloch gauge.
    \begin{enumerate}
        \item If the system enjoys a $\mathfrak{O}^{(++)}$-type symmetry then $\mathrm{wn}([B])\in \mathbb{Z}$.
        \item If the system enjoys a $\mathfrak{O}^{(+-)}$-type symmetry then $\mathrm{wn}([B])\in 2\mathbb{Z}$.
        \item If the system enjoys a $\mathfrak{O}^{(--)}$-type symmetry then $\mathrm{wn}([B])\in 2\mathbb{Z}$.
        \item If the system enjoys a $\mathfrak{O}^{(-+)}$-type symmetry then $\mathrm{wn}([B])\in \mathbb{Z}$.
    \end{enumerate}
\end{Theorem}
\begin{proof}
We will use all the symmetries at our disposal to understand what values $\mathrm{wn}([B])$ can take. We use notions and notation from Appendix~\ref{appendixC} in what follows. First of all we have
\begin{equation*}
\begin{aligned}
\mathrm{wn}([B])&=w(\det(B_k))=w\left(\det(B^{(+)}_k)\det(B^{(-)}_k)\right)=\\
&=w\left(\det(B^{(+)}_k)\right)+w\left(\det(B^{(-)}_k)\right).
\end{aligned}
\end{equation*}
\textit{1.} Let us start with the operator $\mathfrak{O}^{(++)}$; we have
$$\mathfrak{O}^{(++)} v_j(k) = v_{j}(k), \quad \mathrm{and} \quad \mathfrak{O}^{(++)} u_j(k) = u_{j}(k),\quad \forall \ j \in \{1,...,2m\},$$
and 
$$\mathfrak{O}^{(++)} P_k^{(+/-)} = P_k^{(+/-)}\mathfrak{O}^{(++)}.$$
We have
\begin{equation*}
\begin{aligned}
  &B_k^{(-)}\mathfrak{O}^{(++)}v_j(k)=B_k^{(-)}v_j(k)=u_j(k)=\mathfrak{O}^{(++)}u_j(k)=\mathfrak{O}^{(++)}B_k^{(-)}v_j(k);\\
  &B_k^{(+)}\mathfrak{O}^{(++)}v_j(k)=B_k^{(+)}v_j(k)=u_i(k)=\mathfrak{O}^{(++)}u_j(k)=\mathfrak{O}^{(++)}B_k^{(+)}v_j(k);
\end{aligned}    
\end{equation*}
the first one for all $j \in \{1,...,m\}$ and the second one for $j \in \{m+1,...,N\}$. Therefore
\begin{equation*}
\begin{aligned}
&B_k^{(-)}P_k^{(-)}\mathfrak{O}^{(++)}=B_k^{(-)}\mathfrak{O}^{(++)}P_k^{(-)}=\mathfrak{O}^{(++)}B_k^{(-)}P_k^{(-)};\\
&B_k^{(+)}P_k^{(+)}\mathfrak{O}^{(++)}=B_k^{(+)}\mathfrak{O}^{(++)}P_k^{(+)}=\mathfrak{O}^{(++)}B_k^{(+)}P_k^{(+)};
\end{aligned}    
\end{equation*}
from which we deduce that
\begin{equation*}
\begin{aligned}
&P_k^{(-)}B_kP_k^{(-)}\mathfrak{O}^{(++)}=\mathfrak{O}^{(++)}P_k^{(-)}B_kP_k^{(-)} \quad \Rightarrow \quad B_k^{(-)}\mathfrak{O}^{(++)}=\mathfrak{O}^{(++)}B_k^{(-)};\\
&P_k^{(+)}B_kP_k^{(+)}\mathfrak{O}^{(++)}=\mathfrak{O}^{(++)}P_k^{(+)}B_kP_k^{(+)}\quad \Rightarrow \quad B_k^{(+)}\mathfrak{O}^{(++)}=\mathfrak{O}^{(++)}B_k^{(+)}.
\end{aligned}    
\end{equation*}
The above relations, since $\mathfrak{O}^{(++)}$ is linear, do not provide non-trivial information on the winding number of $B^{(\pm)}$.\\
\textit{2.} Let us now consider the case of $\mathfrak{O}^{(+-)}$; in this case 
\begin{equation*}
    \mathfrak{O}^{(+-)}\, v_j(k) = v_{N - j + 1}(k), \quad \mathfrak{O}^{(+-)}\, u_j(k) = u_{N - j + 1}(k) \quad \forall \ j \in \{1,...,m\},
\end{equation*}
and 
$$\mathfrak{O}^{(+-)} P_k^{(+)} = P_k^{(-)}\mathfrak{O}^{(+-)}.$$
We have
\[\begin{split}B_k^{(+)}\mathfrak{O}^{(+-)}v_i(k) &=B_k^{(+)}v_{N-i+1}(k)=u_{N-i+1}(k)=\\ & =\mathfrak{O}^{(+-)}u_i(k)=\mathfrak{O}^{(+-)}B_k^{(-)}v_i(k)\end{split}\]
for all $j \in \{1,...,m\}$.
Therefore
$$B_k^{(+)}P_k^{(+)}\mathfrak{O}^{(+-)}=B_k^{(+)}\mathfrak{O}^{(+-)}P_k^{(-)}=\mathfrak{O}^{(+-)}B_k^{(-)}P_k^{(-)};$$
from which we deduce that
$$P_k^{(+)}B_kP_k^{(+)}\mathfrak{O}^{(+-)}=\mathfrak{O}^{(+-)}P_k^{(-)}B_kP_k^{(-)}\quad \Rightarrow \quad B_k^{(+)}\mathfrak{O}^{(+-)}=\mathfrak{O}^{(+-)}B_k^{(-)}.$$
The above relation, since $\mathfrak{O}^{(+-)}$ is linear, gives
$$\det\left(B_k^{(+)}\right)=\det\left(B_k^{(-)}\right).$$ Therefore we have
$$\mathrm{wn}([B])=2w\left(\det\left(B^{(+)}_k\right)\right)\in 2\mathbb{Z}.$$
\textit{3.} It is now the turn of $\mathfrak{O}^{(--)}$; we have 
\begin{equation*}
    \mathfrak{O}^{(--)} v_j(k) = v_{N - j + 1}(-k), \quad \mathfrak{O}^{(--)} u_j(k) = u_{N - j + 1}(-k) \quad \forall \ j \in \{1,...,m\},
\end{equation*}
and 
$$\mathfrak{O}^{(--)} P_k^{(+)} = P_{-k}^{(-)}\mathfrak{O}^{(--)}.$$
We have
\begin{equation}
\begin{split}B_{-k}^{(+)}\mathfrak{O}^{(--)}v_j(k)&=B_{-k}^{(+)}v_{N-j+1}(-k)=u_{N-j+1}(-k)=\\ & =\mathfrak{O}^{(--)}u_j(k)=\mathfrak{O}^{(--)}B_k^{(-)}v_j(k)\end{split}
\end{equation}
for all $j \in \{1,...,m\}$.
Therefore
$$B_{-k}^{(+)}P_{-k}^{(+)}\mathfrak{O}^{(--)}=B_{-k}^{(+)}\mathfrak{O}^{(--)}P_k^{(-)}=\mathfrak{O}^{(--)}B_k^{(-)}P_k^{(-)};$$
from which we deduce that
$$P_{-k}^{(+)}B_{-k}P_{-k}^{(+)}\mathfrak{O}^{(--)}=\mathfrak{O}^{(--)}P_k^{(-)}B_kP_k^{(-)}\quad \Rightarrow \quad B_{-k}^{(+)}\mathfrak{O}^{(--)}=\mathfrak{O}^{(--)}B_k^{(-)}.$$
The above relation, since $\mathfrak{O}^{(--)}$ is anti-linear, gives
$$\det\left(B_{-k}^{(+)}\right)=\left(\det\left(B_{k}^{(-)}\right)\right)^*=\det\left(\left(B_k^{(-)}\right)^{*}\right)=\det\left(\left(B_k^{(-)}\right)^{-1}\right).$$ Therefore we have again
$$\mathrm{wn}([B])=2w\left(\det\left(B^{(+)}_k\right)\right)\in 2\mathbb{Z},$$
since 
\begin{equation}
\begin{split}
w\left(\det\left(B^{(+)}_k\right)\right)=&-w\left(\det\left(B^{(+)}_{-k}\right)\right)=-w\left(\det\left(\left(B_k^{(-)}\right)^{-1}\right)\right)=\\
&=w\left(\det\left(B^{(-)}_k\right)\right).
\end{split}
\end{equation}
\textit{4.} Last but not least, the operator $\mathfrak{O}^{(-+)}$; we have
$$\mathfrak{O}^{(-+)} v_j(k) = v_{j}(-k), \quad \mathrm{and} \quad \mathfrak{O}^{(-+)} u_j(k) = u_{j}(-k),\quad \forall \ j \in \{1,...,2m\},$$
and 
$$\mathfrak{O}^{(-+)} P_k^{(+/-)} = P_{-k}^{(+/-)}\mathfrak{O}^{(-+)}.$$
We have
\begin{equation*}
\begin{aligned}
  &B_{-k}^{(-)}\mathfrak{O}^{(-+)}v_j(k)=B_{-k}^{(-)}v_j(-k)=u_j(-k)=\mathfrak{O}^{(-+)}u_j(k)=\mathfrak{O}^{(-+)}B_k^{(-)}v_j(k);\\
  &B_{-k}^{(+)}\mathfrak{O}^{(-+)}v_j(k)=B_{-k}^{(+)}v_j(-k)=u_j(-k)=\mathfrak{O}^{(-+)}u_j(k)=\mathfrak{O}^{(-+)}B_k^{(+)}v_j(k);
\end{aligned}    
\end{equation*}
the first one for all $j \in \{1,...,m\}$ and the second one for $j \in \{m+1,...,N\}$. Therefore
\begin{equation*}
\begin{aligned}
&B_{-k}^{(-)}P_{-k}^{(-)}\mathfrak{O}^{(-+)}=B_{-k}^{(-)}\mathfrak{O}^{(-+)}P_k^{(-)}=\mathfrak{O}^{(-+)}B_k^{(-)}P_k^{(-)};\\
&B_{-k}^{(+)}P_{-k}^{(+)}\mathfrak{O}^{(-+)}=B_{-k}^{(+)}\mathfrak{O}^{(-+)}P_k^{(+)}=\mathfrak{O}^{(-+)}B_k^{(+)}P_k^{(+)};
\end{aligned}    
\end{equation*}
from which we deduce that
\begin{equation*}
\begin{aligned}
&P_{-k}^{(-)}B_{-k}P_{-k}^{(-)}\mathfrak{O}^{(-+)}=\mathfrak{O}^{(-+)}P_k^{(-)}B_kP_k^{(-)} \quad \Rightarrow \quad B_{-k}^{(-)}\mathfrak{O}^{(-+)}=\mathfrak{O}^{(-+)}B_k^{(-)};\\
&P_{-k}^{(+)}B_{-k}P_{-k}^{(+)}\mathfrak{O}^{(-+)}=\mathfrak{O}^{(-+)}P_k^{(+)}B_kP_k^{(+)}\quad \Rightarrow \quad B_{-k}^{(+)}\mathfrak{O}^{(-+)}=\mathfrak{O}^{(-+)}B_k^{(+)}.
\end{aligned}    
\end{equation*}
The above relations, since $\mathfrak{O}^{(-+)}$ is anti-linear, give
\begin{equation*}
\begin{aligned}
&\det(B_{-k}^{(+)})=\left(\det(B_{k}^{(+)})\right)^*;\\
&\det(B_{-k}^{(-)})=\left(\det(B_{k}^{(-)})\right)^*.
\end{aligned}    
\end{equation*}
These relations does not link the determinants of $B^{(+)}$ and $B^{(-)}$ and so we have no conditions on $\text{wn}[B]$.

The same result follows in the case $\left(\mathfrak{O}^{(++)}\right)^2=-1_{\mathbb{C}^{2m}}$ and/or the case $\left(\mathfrak{O}^{(-+)}\right)^2=-1_{\mathbb{C}^{2m}}$ by using the linearity of $B_k^{(+)}$ and of $B_k^{(-)}$.
\end{proof}

We can then conclude the following result.
\begin{Corollary}[Gauge invariance of Zak phase's quotients]\label{cor:gauge}\phantom{a}
\begin{enumerate}
     \item If the system enjoys symmetries of type $\mathfrak{O}^{(+-)}$ or $\mathfrak{O}^{(--)}$, then the mod 2 equivalence class
\[
\left[ \frac{1}{2\pi i} \int_{S^1} dk \sum_{j=1}^{2m} \langle v_j(k), \partial_kv_j(k) \rangle \right] 
= \left[ \frac{1}{\pi i} \int_{S^1} dk \sum_{j=1}^{m} \langle v_j(k), \partial_kv_j(k) \rangle \right] \in \mathbb{Z}_2,
\]
is a gauge invariant, independent of the choice of a symmetric Bloch basis.
\item If the system enjoys symmetries of type $\mathfrak{O}^{(++)}$ or $\mathfrak{O}^{(-+)}$ then the mod 1 equivalence class of the Zak phase is a trivial gauge invariant. 
\end{enumerate}
\end{Corollary}
\begin{proof}
    The independence of the mod 2 equivalence of the Zak phase from the Bloch basis is a direct consequence of the above Theorem. Hence the mod-2 equivalence class of the Zak phase is a gauge invariant. The rewriting follows from using the actions of operators $\mathfrak{O}^{(+-)}$ and $\mathfrak{O}^{(--)}$ in Definition \ref{SBB} and manipulating the scalar product in the definition of the Zak phase, cf.~\cite[Corollary IV.8]{Monaco_2023}. Similar considerations hold for the mod-1 equivalence class which contains only the trivial class.
\end{proof}
We therefore can give the following Definition.
\begin{Definition}[The invariant]
    Let $H$ be a finite range Hamiltonian on $\ell^2(\mathbb{Z}) \otimes \mathbb{C}^N$ modeling a translation invariant insulator with time reversal, particle-hole and/or chiral symmetries. Assume the {\rm AZC$-$class} is \emph{not} {\rm A}, {\rm AI} or {\rm AII}. Then the quantity
\[
\mathrm{I}^{(\mathrm{AZC-class})}(H) := \frac{1}{\pi i} \int_{S^1} dk \sum_{j=1}^{m} \langle v_j(k), \partial_kv_j(k) \rangle \mod 2 \in \mathbb{Z}_2,
\]
where is $\{v_i(k)\}_{i=1,\ldots,m}$ is any symmetric Bloch basis, is a gauge invariant of the spectral projection of $H$ onto the negative energy eigenspace.
\end{Definition}
The last general point we discuss is the topological invariance of the $\mathbb{Z}_2$-valued quantity $\mathrm{I}^{(\mathrm{AZC-class})}(H)$. 
\begin{Theorem}[Topological invariance of $\mathrm{I}^{(\mathrm{AZC-class})}(H)$]\label{tiind}
    Assume that \( H(t) \), \( t \in [0,1] \), is a continuous family of bounded Hamiltonian operators on \( \ell^2(\mathbb{Z}) \otimes \mathbb{C}^N \) describing topological insulators in symmetry classes which are not $\mathrm{A,\,  AI,\,  AII}$. Assume moreover that the spectral gap of \( H(t) \) remains open for all \( t \in [0,1] \). Then
\[
\mathrm{I}^{(\mathrm{AZC-class})}(H(0)) = \mathrm{I}^{(\mathrm{AZC-class})}(H(1)).\]
\end{Theorem}
\begin{proof}
The proof can be found in \cite[Theorem IV.10]{Monaco_2023}.
\end{proof}

\section{Symmetry classes with quaternionic structure}\label{quad0}
This Section investigates the possible values that the topological and gauge invariant $\mathrm{I}^{(\mathrm{AZC-class})}(H) \in \mathbb{Z}_2$ can attain. In particular, we will show that, in presence of a symmetry inducing a quaternionic structure on the fiber Hilbert space, the invariant is actually trivialized. Let us start with the following Definition.
\begin{Definition}[Quaternionic structure]\label{quaternionic}
Let $V$ be a complex vector space. If there exist an anti-linear map $O : V \to V$ such that $O^2=-1_{V}$, then $V$ is said to be endowed by a quaternionic structure.
\end{Definition} 
Observe that a complex vector space endowed with a quaternionic structure necessarily has even dimension. Indeed, for any $v \in V$, we claim that \( v \) and \( O(v) \) are linearly independent. Suppose, by contradiction, that \( O(v) = \alpha v \) for some \( \alpha \in \mathbb{C} \). Then
\[
-v = O^2(v) = O(O(v)) = O(\alpha v) = {\alpha}^* O(v) = {\alpha}^* \alpha v = |\alpha|^2 v.
\]
This implies \( |\alpha|^2 = -1 \), which is impossible in~\( \mathbb{C} \). Therefore, \( v \) and \( O(v) \) are linearly independent. We'll write $N=2m$ for the dimension of a vector space endowed with a quaternionic structure.

The first result we need is the following.
\begin{Proposition}[Quaternionic structure and unitary  matrices]\label{det1}
    Let us consider the Hilbert space $\mathbb{C}^{2m}$ endowed with a quaternionic structure, i.e.\ an anti-linear operator $O$ such that \( O^2 = -1_{\mathbb{C}^{2m}} \). Let us assume that $O$ is anti-unitary. Then, a unitary matrix $U$ commuting with \( O \) has spectrum symmetric under complex conjugation, conjugated eigenvalues have the same multiplicity, and its determinant is \( \det( U) = 1 \). 
\end{Proposition}
\begin{proof}
    Let \( U \in \mathrm{U}(2m) \) be a unitary matrix acting on a complex Hilbert space \( \mathbb{C}^{2m} \), and let~\( O \) be an anti-unitary operator such that
\[
O^2 = -1_{\mathbb{C}^{2m}} \qquad \text{and} \qquad OUO^{-1} = U.
\]
Since \( U \) is unitary, all its eigenvalues lie on the complex unit circle. Let \( \lambda \in \mathbb{C} \) be an eigenvalue of \( U \) with eigenvector \( v \in \mathbb{C}^{2m} \), so that
\[
U(v) = \lambda v.
\]
Applying \( O \) to both sides and using anti-linearity:
\[
O(U(v)) = O(\lambda v) = {\lambda}^* O(v).
\]
Since \( O \) commutes with \( U \), we have that
\[
U(O(v)) = O(U(v)) = {\lambda}^* O(v),
\]
so \( O(v) \) is an eigenvector of \( U \) with eigenvalue \( {\lambda}^* \). It follows that the eigenvalues of \( U \) come in complex conjugate pairs with equal multiplicity. Now, let the eigenvalues of \( U \) be \( \lambda_1, {\lambda}^*_1, \ldots, \lambda_m, {\lambda}^*_m \). Then we have
\[
\det (U) = \prod_{j=1}^{m} \lambda_j {\lambda}^*_j = \prod_{j=1}^{m} |\lambda_j|^2 = 1,
\]
since \( |\lambda_j| = 1 \) for all \( j \) because \( U \) is unitary.
\end{proof}
We note, by contrast, that if \( O^2 = +1_{\mathbb{C}^{2m}} \) no such constraint holds: eigenvalues may occur without pairing, and \( \det (U) \) can take any value in \( \mathrm{U}(1) \). Moreover, if $O$ is linear and \( O^2 = -1_{\mathbb{C}^{2m}} \) we get the constraint \(\alpha^2=-1\) or $\alpha = \pm i$, which implies that $v$ and $O(v)$ can be linearly dependent and the conclusion does not hold.

Notice that there is a relation between quaternionic structures and symplectic structures on $\mathbb{C}^{2m}$. Indeed, any antiunitary operator $O$ can be written as $O=J\mathcal{K}$, where $\mathcal{K}$ is the complex conjugation and $J$ is unitary. From $O^2 = -1_{\mathbb{C}^{2m}}$ it follows that $\mathcal{K}J\mathcal{K} = - J^{-1} = - \mathcal{K} J^T J$, i.e. $J = - J^T$ is also skew-symmetric. Correspondingly, unitary matrices commuting with a quaternionic structure $O$ have a connection with symplectic matrices. Indeed, in the notation above, the condition $OUO^{-1}=U$ for a unitary matrix $U$ becomes
\[ J (\mathcal{K}U\mathcal{K}) J^{-1} = U \quad \Longleftrightarrow \quad J = U J U^T\]
which, for matrices $U$ and $J$ with real entries, is the symplecticity condition. Notice how also for symplectic matrices the determinant is necessarily equal to 1 \cite{de2011symplectic}.

We are now ready to state the following fundamental result, which shows how the presence of a quaternionic structure trivializes the $\mathbb{Z}_2$ invariant.
\begin{Theorem}[Zak phase and quaternionic structure]
    Consider a Hamiltonian $H$ that enjoys a $\mathfrak{O}^{(-)}$-type symmetry that squares to $-1_{\mathbb{C}^{2m}}$, where $\mathfrak{O}^{(-)} \in \{\mathfrak{O}^{(-+)}, \mathfrak{O}^{(--)}\}$. Let us endow the Hilbert space $\mathbb{C}^{2m}$ with the quaternionic structure $\mathfrak{O}^{(-)}$. Define $W_{k}:=\mathcal{T}_ke^{-ikX/2\pi}$ as in Theorem~\ref{exbas}. Then $\mathrm{wn}([W]) \in 2 \mathbb{Z}$ and, when defined, $I^{(\mathrm{AZC-class})}(H)=0 \in \mathbb{Z}_2$ . 
\end{Theorem}
\begin{proof}
The proof of Theorem~\ref{exbas} establishes that $W_k$ is unitary and $2\pi$-periodic. Moreover, $W_{\pi}$ commutes with the symmetry operator $\mathfrak{O}^{(-)}$ due to Proposition \ref{propcomm} and $2\pi$-periodicity:
\begin{equation*}
    \mathfrak{O}^{(-)}W_{\pi}=W_{-\pi}\mathfrak{O}^{(-)}=W_{\pi}\mathfrak{O}^{(-)}.
\end{equation*}
By Proposition \ref{det1} we conclude that $\det(W_{\pi})=\det(W_{-\pi})=1$. Since $\mathcal{T}_0 = 1_{\mathbb{C}^{2m}}$, also $\det(W_0)=1$. Let us now compute the winding number $$\mathrm{wn}([W])=w(\det(W_k))=\frac{1}{2\pi i} \int_{S^1} \frac{\partial_k \det (W_k)}{\det (W_k)}dk.$$
Since \( W_k \) is unitary we can write
\[
\det (W_k) = e^{i \theta(k)},
\]
for some smooth real-valued function \( \theta(k) \); hence, by $2\pi$-periodicity of $W_k$ we have
\begin{equation*}
     \theta(k+2\pi)=\theta(k)+2\pi n_1, \qquad  n_1 \in \mathbb{Z}.
\end{equation*} 
Moreover, since $\mathfrak{O}^{(-)}$ is anti-unitary we have
\begin{equation*}
\det (W_{-k}) = \det \left(\mathfrak{O}^{(-)} W_k \left(\mathfrak{O}^{(-)}\right)^{-1}\right) = \left({\det (W_k)}\right)^*= e^{-i \theta(k)},
\end{equation*} 
so that
\begin{equation*}
\theta(-k) = -\theta(k) + 2\pi n_2, \qquad  n_2 \in \mathbb{Z};
\end{equation*} 
differentiating both sides gives:
\begin{equation*}
-\theta'(-k) = -\theta'(k),
\end{equation*} 
which implies that \( \theta'(k) \) is an even function.\\
Due to $2\pi$-periodicity of $\theta'(k)$ we can compute $w(\det(W_k))$ in the interval $[-\pi,+\pi]$ as
\begin{equation*}
    w(\det(W_k))=\frac{1}{2\pi} \int_{-\pi}^{+\pi} \theta'(k)  dk
\end{equation*}
and, noting that by parity of $\theta'(k)$
\begin{equation*}
    \int_{-\pi}^{0} \theta'(k)  dk=-\int_{\pi}^0 \theta'(-\kappa)d\kappa=\int_0^{\pi} \theta'(-\kappa)d\kappa=\int_0^{\pi} \theta'(\kappa)d\kappa,
\end{equation*}
we conclude that
\begin{equation*}
    w(\det(W_k))=\frac{1}{2\pi} \int_{-\pi}^{+\pi} \theta'(k)  dk=2\left(\frac{1}{2\pi} \int_{0}^{+\pi} \theta'(k)  dk\right) \in 2\mathbb{Z}
\end{equation*}
since 
\begin{equation*}
    \frac{1}{2\pi} \int_{0}^{+\pi} \theta'(k)  dk=\frac{1}{2\pi i } \int_{0}^{+\pi} \frac{\partial_k \det (W_k)}{\det (W_k)} dk
\end{equation*}
is a winding number in half circumference due to the condition $\det(W_0)=\det(W_{\pi})$.
\end{proof}
Therefore, whenever the physical system admits a symmetry represented, at the level of Hilbert space, by an anti-linear operator such that its square is $-1_{\mathbb{C}^{2m}}$, the Zak phase is even and the previously defined $\mathbb{Z}_2$ invariant turns out to be identically zero. In this case, the invariant does not read topological triviality but instead the presence of a particular structure: the quaternionic one. It follows that the $\mathbb{Z}_2$ invariant is actually influenced also by non-topological structures; only when a quaternionic structure is not present, therefore the invariant is not identically zero, it contains topological information that reports a trivial or non-trivial topological phase. Note, however, that the vanishing of $\mathrm{I}^{(\mathrm{AZC-class})}(H)$ does not necessarily imply the presence of a quaternonic structure. In fact, if $\mathrm{I}^{(\mathrm{AZC-class})}(H)$ is zero, there are two possibilities: either there is a quaternionic structure or the system is in the topologically trivial phase. Since the symmetry operators of the theory are generally known, it is relatively simple to discard the hypothesis of the presence of a quaternionic structure. Conversely, as already noted, if a quaternionic structure is present, the invariant $\mathrm{I}^{(\mathrm{AZC-class})}(H)$ loses its ability to report the presence of a non-trivial topological phase. In conclusion, the only AZC symmetry classes where the invariant can attain non-trivial values are classes AIII, BDI and D.

\section{Focus on the symmetry class BDI}\label{focBDI} 
In order to provide an example of a (partial) topological information captured by the Zak phase invariant, we focus on a fully symmetric chain in class BDI, i.e.\ a model which displays all three discrete symmetry operators discussed in Appendix \ref{3.2}, which in turn square to the identity. 
In this case, the invariant $\mathrm{I}^{(\mathrm{BDI})}(H)$ records ``a bit'' of the topology of the system; in fact, Kitaev's classification labels each BDI chain with a $\mathbb{Z}$-valued invariant, while the Zak phase returns the invariant $\mathrm{I}^{(\mathrm{BDI})}(H) \in \mathbb{Z}_2$. The aim of this section is to show, at least within an explicit but fairly general example, that the invariant $\mathrm{I}^{(\mathrm{BDI})}(H)$ takes into account the parity of the topological invariant predicted by Kitaev's classification.
\subsection{Arbitrary distant nearest neighbors Kitaev chain}
We consider a generalization of Kitaev's chain, originally proposed in \cite{kitaev2001} as simplified model for a topological superconductor. Kitaev's Hamiltonian accounts only for nearest-neighbour hoppings. In our generalization we consider instead an Hamiltonian $H$ with fibered Hamiltonians given as follows:
\begin{equation*}
    H_k=\begin{bmatrix}
0 & \sum_{n \in \mathbb{Z}}c_ne^{ink} \\
\sum_{n \in \mathbb{Z}}c^*_ne^{-ink} & 0
\end{bmatrix}  \in \mathrm{M}_2(\mathbb{C}).
\end{equation*}
On the fiber Hilbert space $\mathbb{C}^2$, the symmetry operators are given by
\begin{itemize}
  \item time reversal symmetry $\mathfrak{T} = \mathcal{K}$ (complex conjugation);
  \item particle-hole symmetry $\mathfrak{C} = \sigma_3 \mathcal{K}$, where $\sigma_3 = \begin{bmatrix} 1 & 0 \\ 0 & -1 \end{bmatrix}$ is the third Pauli matrix ;
  \item chiral symmetry  $\mathfrak{S}=\sigma_3 = \mathfrak{T}\mathfrak{C}$.
\end{itemize}
Requiring the appropriate (anti)commutation relations of $H_k$ with such symmetry operators imposes the conditions $c_n^*=c_n$. This model takes into account interactions with arbitrarily distant nearest neighbors and we can impose $c_n=0$ if $|n|>R$ to fit the condition of finite-range hopping model. If we set $z(k):=\sum_{n \in \mathbb{Z}}c_ne^{ink}$, the $\mathbb{Z}$-valued invariant predicted by Kitaev's periodic table can be computed as the winding number around the origin of the complex curve $\Gamma$ obtained as the image of map $z \colon S^1 \to \mathbb{C}$ \cite{graf2018bulk}.

The Bloch bands are given by
\[ E_{\pm}(k) =\pm \sqrt{z(k)z(k)^*}=\pm |z(k)|. \]
Solving the eigenvalue equation \( H_kv(k) = E_-(k)v(k) \) yields the normalized eigenvector
\begin{equation} \label{eqn:KitaevEigenv}
v(k) = \frac{1}{\sqrt{2}} 
\begin{bmatrix}
-1 \\
\frac{z^*(k)}{|z(k)|}
\end{bmatrix},
\end{equation}
from which we find, with a straightforward computation,
\[
\langle v(k), \partial_k v(k) \rangle = \frac{1}{4} \frac{(\partial_k z(k))^*z(k)-z^*(k)(\partial_k z(k))}{z^*(k)z(k)}= -\frac{i}{2} \text{Im}\left( \frac{\partial_kz(k)}{z(k)} \right).
\]
Therefore, the Zak phase associated to the choice \eqref{eqn:KitaevEigenv} of the eigenvector for the negative band --- or rather to the choice $\{v(k), \mathfrak{S}(v(k))\}$ of a Bloch basis for $\mathbb{C}^2$ --- reads (compare Corollary \ref{cor:gauge})
\begin{align*}
 \mathcal{Z}^{(2)}_{\{v(k), \mathfrak{S}(v(k))\}} & = \frac{1}{\pi i} \int_{S^1} \langle v(k), \partial_k v(k) \rangle dk = - \frac{1}{2\pi} \int_{S^1} \text{Im}\left( \frac{\partial_kz(k)}{z(k)} \right)= \\
 & = \text{Re}\left(- \frac{1}{2\pi i} \int_{S^1} \frac{\partial_kz(k)}{z(k)} \right).
\end{align*}
By the Cauchy integral formula, the integral on the right-hand side of the above computes the negative of the winding number $w(\Gamma)$ around the origin of the curve $\Gamma$: this winding number is an integer,
and therefore the real part is redundant. 
We conclude that in the generalized Kitaev chain the \(\mathbb{Z}_2\) invariant can be computed as
\[
\mathrm{I}^{(\mathrm{BDI})}(H) = w(\Gamma) \bmod 2
\]
and thus the Zak phase can only account for the parity of the $\mathbb{Z}$ invariant predicted by $K$-theory.

\subsection{Multi-channel arbitrary distant nearest neighbors Kitaev chain}
A further way to generalize Kitaev's chain is to allow $2m$ internal degrees of freedom per unit cell (the case with $m$ ``chiral channels'') and hopping/pairing terms with finite (but arbitrarily large) range. In Bloch form, we consider the fibered Hamiltonians
\begin{equation*}
    H_k=
    \begin{bmatrix}
        0 & A_k \\
        A_k^{\dagger} & 0
    \end{bmatrix} \, ,
\end{equation*}
where $A_k \in \mathrm{M}_{m}(\mathbb{C})$ is a matrix-valued trigonometric polynomial of the form
\begin{equation*}
    A_k := \sum_{n\in\mathbb{Z}} A_n e^{i n k} \, , 
    \qquad
    A_n = 0 \quad  \text{if } |n|>R \, .
\end{equation*}
We choose the same symmetry operators as in the scalar case, now promoted to maps on~$\mathbb{C}^{2m}$, namely
\begin{equation*}
    \mathfrak{T}=\mathcal{K} \, , 
    \qquad 
    \mathfrak{C}=\Sigma_3 \mathcal{K} \, , 
    \qquad 
    \mathfrak{S}=\Sigma_3 \, ,
    \qquad
    \Sigma_3 :=
    \begin{bmatrix}
        1_{\mathbb{C}^m} & 0\\
        0 & -1_{\mathbb{C}^m}
    \end{bmatrix} \, .
\end{equation*}
The symmetries operators impose
\begin{equation*}
    A_{-k} = A_k^*,
\end{equation*}
equivalently $A_n^*=A_n$ for all $n$ (i.e all Fourier coefficients are real matrices).

The spectrum of $H_k$ is symmetric with respect to $0$ and is determined by the singular values of $A_k$. More precisely, if $\{\lambda_j(k)\}_{j=1}^m$ are the (non-negative) eigenvalues of the matrix $A_k^{\dagger}A_k$, then
\begin{equation*}
    E_{\pm,j}(k)=\pm \sqrt{\lambda_j(k)} \, , 
    \qquad j\in \{1,\dots,m\} \, .
\end{equation*}
Assuming that the system is gapped at $0$, i.e. $    \det(A_k)\neq 0 \ \text{for all } k\in S^1,$
then $A_k$ admits a smooth polar decomposition
\begin{equation*}
    A_k = U_k F_k \, ,
    \qquad
    U_k := A_k (A_k^{\dagger}A_k)^{-1/2} \in \mathrm{U}(m) \, ,
    \qquad
    F_k := (A_k^{\dagger}A_k)^{1/2} > 0 \, .
\end{equation*}
As in the scalar case, the occupied (negative-energy) subspace can be described by a smooth frame. A convenient choice is the $2m\times m$ matrix
\begin{equation} \label{eqn:Vk}
    V(k)
    :=
    \frac{1}{\sqrt{2}}
    \begin{bmatrix}
        -1_{\mathbb{C}^m}\\
        U_k^{\dagger}
    \end{bmatrix} \, ,
    \qquad
    V(k)^{\dagger}V(k)=1_{\mathbb{C}^m} \, ,
\end{equation}
which is an orthonormal frame of the occupied subspace and it reduces to $\frac{1}{\sqrt{2}}(-1,\, z^*/|z|)^T$ when $m=1$, compare \eqref{eqn:KitaevEigenv}. A direct computation, similar to the one performed in the previous section, exhibits the Zak phase invariant as
\begin{equation*}
\begin{aligned}
    \mathrm{I}^{(\mathrm{BDI})}(H)
    &=\left[\frac{1}{\pi i}\int_{S^1}\mathrm{tr}\,\bigl(V(k)^{\dagger}\partial_k V(k)\bigr)\,dk\right]=\left[-\frac{1}{2\pi i}\int_{S^1}\,\mathrm{tr}\,\bigl(U_k^{\dagger}\partial_k U_k\bigr)\,dk\right]=\\
    &=\left[-\frac{1}{2\pi i}\int_{S^1}\frac{\partial_k\det(U_k)}{\det(U_k)}\,dk\right] \, .
\end{aligned}
\end{equation*}
The integral computes the winding number of the loop $k\mapsto \det(U_k)$ around the origin, hence it is an integer and the brackets read its parity. Finally, the winding can be written directly in terms of $A_k$. Indeed, from $A_k=U_kF_k$ and $\det(F_k)>0$ for all $k$, we have
\begin{equation*}
    \det(U_k)=\frac{\det(A_k)}{|\det(A_k)|} \, ,
\end{equation*}
so the loop $\det(U_k)$ has the same winding number as $\det(A_k)$. Indeed, the normalization mapping $z \mapsto \frac{z}{|z|}$
represents a continuous homotopy that radially contracts the  determinant's values onto the unit circle; as long as $\det(A_k) \neq 0$ for all $k$ (i.e. the system remains gapped and no phase transitions occur), this contraction never crosses the origin and does not alter the total phase accumulated during the loop. Consequently, the map $k \mapsto \det(A_k)$ and its normalized counterpart 
$k \mapsto \det(U_k)$ belong to the same homotopy class; therefore, their 
winding numbers must be same:
\begin{equation*}
    w[\det(A_k)] = w[\det(U_k)]
\end{equation*}

In conclusion, the invariant
$\mathrm{I}^{(\mathrm{BDI})}(H)$ reads the parity of the winding number of the complex function $k\mapsto \det(A_k)$, that is $\text{wn}([A_k])$ (compare again \cite{graf2018bulk}). This computation also suggest that, although the Zak phase is generically not gauge invariant, it is still possible to compute the $\mathbb{Z}$-valued invariant by means of an appropriate, specific symmetric and periodic frame, like the one exhibited in \eqref{eqn:Vk}.

\section{Conclusion and outlook}
The study of topological phases of matter has unveiled deep and unexpected connections between quantum physics, geometry and topology. Along this path tools from functional analysis, bundle theory and $K$-theory have played a fundamental role in bridging concepts from pure mathematics with physical phenomena in condensed matter physics. Within this context, our main goal has been to understand whether the Zak phase of a condensed matter system contains topological information and how complete this information is compared to the ten-fold way classification.

We have investigated chains from all AZC symmetry classes and with finitely many internal degrees of freedom, extending the results of \cite{Monaco_2023}. We have shown how a $\mathbb{Z}_2$-valued gauge and topological invariant, built using the Zak phase, can be naively defined for every AZC class in $1$D. However, in this definition of the invariant $\mathrm{I}^{(\mathrm{AZC-class})}(H)$ one does not consider the characterizing features of each symmetry class, that is, which discrete symmetry operators square to $\pm 1_{\mathbb{C}^{2m}}$. Completing the analysis in~\cite{Monaco_2023}, we showed how in presence of a quaternionic structure, i.e. an anti-linear operator squaring to $-1_{\mathbb{C}^{2m}}$, the invariant $\mathrm{I}^{(\mathrm{AZC-class})}(H)$ has to vanish. In this case, the invariant does not read the topologically trivial phase but instead the presence of a particular structure: the quaternionic one. Among the symmetry classes in which the $\mathbb{Z}_2$ invariant can be non-zero, we focused on the BDI class, which the ten-fold way classification labels with a $\mathbb{Z}$ invariant. We therefore analyzed generalized Kitaev chains with arbitrary range couplings in class BDI and we underlined how $\mathrm{I}^{(\mathrm{BDI})}(H)$ can be interpreted as the parity of the $\mathbb{Z}$ invariant provided by the periodic table. 

A natural direction for future research concerns the extension of the present analysis to higher spatial dimensions. In this setting, it would be particularly interesting to investigate whether Zak phase-type invariants associated with one-dimensional subsystems, commonly interpreted as weak invariants, retain a comparable sensitivity to additional geometric structures, such as quaternionic structures, in two- and three-dimensional systems with full discrete symmetry constraints. More generally, the analysis of complete topological invariants in lattice-periodic systems, replicating and complementing those in Kitaev's table with weak invariants, is a challenging task: this work was initiated in \cite{PelusoThesis} for symmetry classes displaying a single simmetry operator, and a general study for fully-symmetric classes is ongoing \cite{MMP}. From a more physical perspective, it would also be of interest to analyze the robustness of the proposed invariant under perturbations breaking strict translation invariance, such as disorder, or in the presence of weak interactions, in order to assess to what extent the construction survives beyond the idealized non-interacting, perfectly periodic regime: stable indices of this sort have been proposed in \cite{grossmann2016index, thiang2016k, chung2025topological, Chung_Shapiro_II, Chung_Shapiro_All, bachmann2020many, bachmann2025many, bachmann2026index} with a range of different mathematical settings and techniques, also for higher-dimensional materials. 

\appendix 

\section{Bundles from condensed matter systems}\label{2.1}

In this Appendix, we recall the main notions to formulate a mathematical model of a crystalline solid-state system. In view of the discussion in the main text, we will formulate the theory for discrete, tight-binding Hamiltonians; with further analytical care, the setup can be extended to include continuum models of solids, using Hamiltonian of differential type (e.g.\ Schr\"odinger operators). We refer the reader to~\cite{ReedSimonIV, kuchment2016overview, Monaco_2014, Monaco_2018, lewin2024spectral} and references therein for more detailed accounts of this mathematical framework.

\subsection{Lattice Hamiltonians}

In solids, the ionic cores, whose positions can be assumed to be fixed in view of the Born-Oppenheimer approximation \cite{bransden2003physics}, form a crystal, that is, a regular arrangement in space which repeats in all directions. Electrons move in this background. To model its quantum dynamics, we need a Hamiltonian which is periodic with respect to a translation group. This group is called the Bravais lattice of translations $\Gamma$: in $d$-dimensions,
\begin{equation*}
    \Gamma:=\mathrm{Span}_{\mathbb{Z}}\{a_1,\ldots,a_d\}\simeq \mathbb{Z}^d \subset \mathbb{R}^d, 
\end{equation*}
where $\{a_1,\ldots,a_d\}$ is a basis of $\mathbb{R}^d$, in general different from the regular arrangement of positions of the ionic cores.
The periodicity condition means that translations by vectors $\gamma \in \Gamma$ are implemented as operators $U_{\gamma}$ acting unitarily on the Hilbert space $\mathcal{H}$ of the electron as dynamical symmetries, i.e. 
\begin{equation*}
    [U_{\gamma},H]=0.
\end{equation*}
Moreover, the map $U: \Gamma \rightarrow \mathcal{U}(\mathcal{H})$ that maps $\gamma \mapsto U_{\gamma}$ must be a unitary representation of $\Gamma$ on $\mathcal{H}$. For $\mathcal{H}=\ell^2(\Gamma) \otimes \mathbb{C}^N \simeq \ell^2(\Gamma;\mathbb{C}^N)$, for example, translations are implemented as 
\begin{equation*}
    (U_{\gamma}(\psi))_n=\psi_{n+\gamma}, \qquad \forall \ \psi = (\psi_n)_{n \in \Gamma} \in \mathcal{H}, \ \gamma \in \Gamma\,.
\end{equation*}
Such discrete models often occur in condensed matter physics: each cell of the Bravais lattice contains only a finite number $N \in \mathbb{N}$ of degrees of freedom. Electrons can then only jump among lattice sites, as if they were bound to their positions: this is the so-called tight-binding approximation \cite{TB1,TB2}. The typical action of the Hamiltonian on $\mathcal{H}$ is given by
\begin{equation*} \label{eq:tight_binding}
(H(\psi))_x = \sum_{y \in \Gamma} t_{x,y} \psi_y + V_x \psi_x, \qquad \psi = (\psi_y)_{y \in \Gamma} \in \mathcal{H},\ x \in \Gamma.
\end{equation*}
The coefficient $t_{x,y}$ is interpreted as the ($N \times N$-matrix-valued) ``hopping amplitude'' for the electron to jump from site $x$ to site $y$, while $V_x$ represents an on-site potential energy. Commutativity with lattice translations require $t_{x,y} \equiv t_{x-y}$ for all $x,y \in \Gamma$.

Let us note that the operators $U_{\gamma}$ commute since $\Gamma$ is an abelian group. Hence, we can simultaneously diagonalize the operators $U_{\gamma}$ and we can choose a common
generalized eigenvector $\psi$ for all the translation operators such that
\begin{equation*}
    U_{\gamma}(\psi) = \chi_{\gamma}\, \psi \qquad \forall \ \gamma \in \Gamma
\end{equation*}
where $\chi_{\gamma} \in \mathrm{U}(1)$ since $U_{\gamma}$ is a unitary operator. It is possible to show that also $\chi : \Gamma \rightarrow \mathrm{U}(1)$, mapping $\gamma \mapsto \chi_{\gamma}$, is a unitary representation of $\Gamma$; this map is in fact a unitary character. Since $\Gamma \simeq \mathbb{Z}^d$ we must have 
\begin{equation*}
    \chi_{\gamma} \equiv \chi_\gamma(k) :=e^{ik \cdot \gamma}, \qquad k \in \mathbb{R}^d/\Gamma^*,
\end{equation*}
where $\Gamma^*$ is the reciprocal lattice 
\begin{equation*}
    \Gamma^*:=\{\lambda \in \mathbb{R}^d \ | \ \lambda \cdot \gamma \in 2\pi \mathbb{Z} \quad \forall \ \gamma \in \Gamma  \} \simeq (2\pi \mathbb{Z})^d.
\end{equation*}
Therefore
\begin{equation*}
    \Gamma^*=\text{Span}_{2\pi \mathbb{Z}}\{b_1,...,b_d\},
\end{equation*}
where $\{b_1,...,b_d\}$ is the dual basis defined by $a_{i}\cdot b_{j}=\delta_{i,j} \ i,j \in \{1,...,d\}.$ The quotient $\mathbb{R}^d/\Gamma^* \simeq \mathbb{T}^d$ is know as the Brillouin torus in condensed matter literature after the pioneering work \cite{brill}. The parameter $k \in \mathbb{T}^d$, which replaces the linear momentum\footnote{The linear momentum is a good quantum number in the case of a quantum system in which the translational symmetry group is the full $(\mathbb{R}^d,+).$}, is called Bloch momentum or  quasi-momentum. We note that, writing $k \in \mathbb{R}^d/\Gamma^*$ simply means that $k \in \mathbb{R}^d$ is determined up to elements in
$\Gamma^*$, or equivalently that $\chi_{\gamma}$ is $\Gamma^*$-periodic, i.e. $e^{i(k+\lambda)\cdot \gamma}=e^{ik\cdot \gamma}$ for all $\lambda \in \Gamma^*$.

We can now look at how the Hamiltonian $H$ acts on wavefunctions $\psi_k$ with well-defined quasi-momentum $k$: since the Hamiltonian $H$ and the translation operators $U_{\gamma}$ commute, it will map such a wavefunction to one with the same quasi-momentum. Therefore, we can formulate the energy eigenvalue problem for fixed $k \in \mathbb{T}^d$
\begin{equation}\label{eqbh}
    H_k(\psi_k) =E_k\psi_k,
\end{equation}
with the additional condition 
\begin{equation}\label{ctra}
    U_{\gamma}(\psi_k)=\chi_{\gamma}(k) \psi_k=e^{ik \cdot \gamma}\psi_k.
\end{equation}

We give the following useful definitions.
\begin{Definition}[Bloch notions]
Referring to \eqref{eqbh} and \eqref{ctra}, the generalized eigenfunction $\psi_k$ are known as Bloch function and the corresponding energies $E_k$, viewed as
functions of $k$, are called Bloch bands. Finally, the operator $H_k$, which acts as the Hamiltonian but only on wavefunctions labeled by the quasi-momentum $k$, is called the Bloch Hamiltonian.
\end{Definition}
 
Functions in~\eqref{ctra} constitute a $k$-dependent Hilbert space, called the fiber Hilbert space: for $\mathcal{H}=\ell^2(\Gamma;\mathbb{C}^N)$, this reads 
\begin{equation}\label{hilbk}
    \mathcal{H}_k:=\{\psi \in \ell_{\text{loc}}^2(\Gamma;\mathbb{C}^N) \ | \ U_{\gamma} (\psi)=e^{ik\cdot \gamma} \psi, \ \gamma \in \Gamma\} \simeq \mathbb{C}^N\,,
\end{equation}
where the last isomorphism is realized by evaluation of the restriction of any function satisfying the pseudo-periodic boundary condition~\eqref{ctra} to a unit cell of the Bravais lattice $\Gamma$ (the number of points in this cell possibly contributes to the $N$ local degrees of freedom); the inverse is given by pseudo-periodic extension from the fundamental cell to the whole $\Gamma$. In particular, the Bloch Hamiltonians $H_k$, $k \in \mathbb{T}^d$, can be realized as $N \times N$ Hermitian matrices.

\subsection{Bloch Hamiltonians and the band-gap theory}\label{sub22.2}

One way to ``extract'' the Bloch Hamiltonians $(H_k)_{k \in \mathbb{T}^d}$ from $H$ is through the Bloch--Floquet transform~\eqref{tra}. The Hilbert space which constitutes the range of the Bloch--Floquet transform $\mathcal{F}_d$ is isomorphic to the direct integral Hilbert space $\int_{\mathbb{T}^d}^{\oplus}\mathcal{H}_k\,dk$ parameterized by the measure space given by $(\mathbb{T}^d,\Omega,dk)$. A particular class of operators on this Hilbert space are the so-called fibered operators.
\begin{Definition}[Fibered operator]
 A fibered operator is an operator of the form 
 \begin{equation*}
    O=\int_{\mathbb{T}^d}^{\oplus}O_kdk,
\end{equation*}
where $O_k$ is a family of operators on $\mathcal{H}_k$ parametrized by $k \in \mathbb{T}^d$ which
acts fiberwise on the direct integral and such that the map $\Phi: k \mapsto O_k$ is measurable in the sense that the functions $$\Phi_{ij} : k \mapsto \langle f_i, O_k f_j\rangle $$
are measurable for every $i,j \in \{1,\ldots, N\}$, where $\{f_j\}_{j=1}^{N}$ is the canonical basis of $\mathbb{C}^N \simeq \mathcal{H}_k$ (compare~\eqref{hilbk}).
\end{Definition}
The following Proposition shows that the translation-invariant Hamiltonian $H$ leads to a fibered operator.
\begin{Proposition}
    Let $A$ be an operator on $\mathcal{H}$ such that $[U_{\gamma},A]=0$. Then the transformed operator $O=\mathcal{F}_d A \mathcal{F}_{d}^{-1}$ is a fibered operator. Moreover, the fiber operators $\{O_k\}_{k \in \mathbb{T}^d}$ fulfill $\Gamma^*$-periodicity $$O_{k+\lambda}=O_k, \qquad \lambda \in \Gamma^*.$$ 
\end{Proposition}
\begin{proof}
The vectors $\{e_\gamma \otimes f_i\}_{\gamma \in \Gamma,\ i \in \{1,\ldots,N\}}$ constitute a basis for $\mathcal{H} = \ell^2(\Gamma) \otimes \mathbb{C}^N$. Set
\[ A_{\beta, \gamma; i, j} := \left\langle e_\beta \otimes f_j,\ A (e_\gamma \otimes f_i) \right \rangle\,, \quad \beta, \gamma \in \Gamma,\ i,j \in \{1,\ldots,N\}, \]
and observe that
\begin{align*}
A_{\beta, \gamma; i, j} & = \left\langle U_\beta \, e_0 \otimes f_j,\ A (e_\gamma \otimes f_i) \right \rangle = \left\langle e_0 \otimes f_j,\ U_{-\beta} A (e_\gamma \otimes f_i) \right \rangle \\
& = \left\langle e_0 \otimes f_j,\ A U_{-\beta} (e_\gamma \otimes f_i) \right \rangle = \left\langle e_0 \otimes f_j,\ A (e_{\gamma-\beta} \otimes f_i) \right \rangle \\
& = A_{0,\gamma-\beta; i,j}\,.
\end{align*}
Consequently
\begin{align*}
\mathcal{F}_d A (e_\gamma \otimes f_i)& = \sum_{\substack{\beta \in \Gamma\\1 \le j \le N}} A_{\beta, \gamma; i, j} \, \mathcal{F}_d (e_\beta \otimes f_j) = \sum_{\substack{\beta \in \Gamma\\1 \le j \le N}} A_{0,\gamma-\beta; i, j} \, \frac{e^{ik\cdot\beta}}{\sqrt{2\pi}}\, f_j \\
&= \frac{e^{ik\cdot\gamma}}{\sqrt{2\pi}} \, \sum_{\substack{\eta \in \Gamma\\1 \le j \le N}} A_{0,\eta; i, j} \, e^{-ik\cdot\eta} \, f_j\,.
\end{align*}

Define now
\[ O_k \colon \mathcal{H}_k \to \mathcal{H}_k \,, \quad f_i \mapsto \sum_{\substack{\eta \in \Gamma\\1 \le j \le N}} A_{0,\eta; i, j} \, e^{-ik\cdot\eta} \, f_j\,. \]
It is clear then that, by construction
\[ \mathcal{F}_d A (e_\gamma \otimes f_i) = O_k \mathcal{F}_d (e_\gamma \otimes f_i) \]
and thus that $(O_k)_{k \in \mathbb{T}^d}$ are the fibers of $O = \mathcal{F}_d A \mathcal{F}_d^{-1}$. Upon inspection, since $\lambda \cdot \eta \in 2\pi\mathbb{Z}$ for all $\lambda \in \Gamma^*$, one also immediately gets that $O_{k+\lambda} = O_k$ for all such $\lambda$'s, as wanted.
\end{proof}

An immediate application of the above Proposition is that
\begin{equation*}
    \mathcal{F}_dH\mathcal{F}_d^{-1}=\int_{\mathbb{T}^d}^{\oplus} H_k dk
\end{equation*}
is a fibered operator, and the fiber Hamiltonians $H_k$ are self-adjoint on $\mathcal{H}_k$. Denoting their eigenvalues (Bloch bands) as $E_{n,k}$, assumed to be labeled in non-decreasing order
\begin{equation*}
    E_{0,k}\leq E_{1,k} \leq ... \leq E_{s,k}\leq E_{s+1,k}\leq ... \leq +\infty, 
\end{equation*}
the total spectrum of $H$ can be reconstructed as the union of the Bloch
bands
\begin{equation*}
    \sigma(H)=\bigcup_{n \in \{1,\ldots,N\}}\bigcup_{k \in \mathbb{T}^d}\{E_{n,k}\}.
\end{equation*}
The image of the $\Gamma^*$-periodic function $E_{n,k}: \mathbb{T}^d \rightarrow \mathbb{R}$ constitutes a
part of the spectrum of $H$. Different Bloch bands may have overlapping images, and their union in the spectrum of $H$ is referred to as a spectral band. When these spectral bands do not overlap, they are said to be separated by a spectral gap. As a result, the spectrum of $H$ alternates between spectral bands and spectral gaps, forming what is known as the band-gap structure. From a physical point of view, band-gap structure theory has been successfully used to explain many physical properties of solids, such as electrical resistivity, electrical conductivity, optical absorption and it is the foundation of the understanding of all solid-state devices like transistors or solar cells. 

\subsection{The Bloch bundle}

Under the spectral gap hypothesis that $H_k$ has no spectrum around $\mu = 0$ for all $k \in \mathbb{T}^d$, we have shown how the Riesz formula \eqref{eq:riesz} defines the spectral eigenprojections $P_k \equiv P^{(-)}_k$ onto negative energy bands. Assume that there are $m \in \{1, \ldots, N\}$ such bands, so that the rank of the projections is $m$. These data can be recast into a vector bundle over the Brillouin torus $\mathbb{T}^d$. 
\begin{Definition}[Bloch bundle]\label{blochbu}
The Bloch bundle associated to a condensed matter system is the triple $\mathcal{E}_{\mathrm{Bloch}}=(E_{\mathrm{Bloch}}, \pi_{\mathrm{Bloch}},X_{\mathrm{Bloch}})$ such that:
\begin{enumerate}
    \item the base $X_{\mathrm{Bloch}}$ is the Brillouin torus $\mathbb{T}^d:=\mathbb{R}^d/\Gamma^*$;
    \item the total space $E_{\mathrm{Bloch}}$ is defined as
    $$  E_{\mathrm{Bloch}}:=\{[k,\phi_k] \in (\mathbb{R}^d\times \mathcal{H}_k)/\sim_{\lambda} \ | \ \phi_k \in \mathrm{Range}(P_k)\}$$
    where $$ (k,\phi_k) \sim_{\lambda} (k',\phi'_{k'}) \qquad  \mathrm{if} \quad k'=k+\lambda, \ \phi'_{k'}=\tau_{\lambda}(\phi_k),$$
    for some $\lambda \in \Gamma^*$;
    \item the projection $\pi_{\mathrm{Bloch}}$ is given by $$\pi_{\mathrm{Bloch}} :  E_{\mathrm{Bloch}} \rightarrow \mathbb{T}^d, \quad [k,\phi] \mapsto k \ \mathrm{mod} \ \Gamma^*.$$
\end{enumerate}
\end{Definition}
The next Proposition shows that the Bloch bundle is in fact a complex vector bundle of rank $m$. The proof can be found e.g.\ in~\cite{Monaco_2014}.
\begin{Proposition}
    The triple $\mathcal{E}_{\mathrm{Bloch}}=(E_{\mathrm{Bloch}}, \pi_{\mathrm{Bloch}},X_{\mathrm{Bloch}})$ defined above is a complex vector bundle of rank $m$ over the torus $\mathbb{T}^d$.
\end{Proposition}

\section{Time reversal, particle-hole and chiral symmetries}\label{3.2}
\subsubsection*{Time reversal symmetry}
A time reversal symmetry is a projective anti-unitary representation on the Hilbert space of the quantum system $\mathcal{H}$ of the group \( (\mathbb{Z}_2 = \{ \pm 1 \},\cdot) \), specified by the action of an operator \( T \)  that represent the non-trivial element \( -1 \in \mathbb{Z}_2 \). It is required to be a dynamical symmetry of the system, that is, it commutes with the Hamiltonian \( H \), since time reversal symmetry can be read off directly from the Schr\"{o}dinger equation. Being a projective anti-unitary representation, we have that
\begin{equation*}
T^2 = e^{i\vartheta} 1_{\mathcal{H}}.
\end{equation*}
By associativity and the fact that \( T \) is anti-unitary, we conclude that
\begin{equation*}
 \begin{aligned}  
&T^3 = T^2 T = e^{i\vartheta} T, \\ 
&T^3=T T^2 = T e^{i\vartheta} = e^{-i\vartheta} T,
\end{aligned} 
\end{equation*}
that is,
\begin{equation*}
e^{i\vartheta} = e^{-i\vartheta},
\end{equation*}
which implies that \( e^{i\vartheta} \in \{ \pm 1 \} \). A time reversal operator such that \( T^2 = +1_{\mathcal{H}} \) is called even while one such that \( T^2 = -1_{\mathcal{H}} \) is called odd.

\subsubsection*{Particle-hole symmetry}\label{APPB}

In quantum mechanics, so in particular in condensed matter physics, Dirac's theory implies a pairing between particles and holes (i.e. conducting and valence electrons in a solid). The definition of what is a ``particle'' and what is a ``hole'' is in fact arbitrary, hence the theory is required to be invariant under their exchange. This charge conjugation, or particle-hole symmetry, is described as a projective anti-unitary representation of \( (\mathbb{Z}_2,\cdot) \), specified by a charge conjugation operator \( C \) representing \( -1 \in \mathbb{Z}_2 \). Two types of representations are possible depending on whether
\begin{equation*}
C^2 = +1_{\mathcal{H}} \quad \text{or} \quad C^2 = -1_{\mathcal{H}}.
\end{equation*}
However, unlike time reversal, charge conjugation is not a dynamical symmetry in a strict sense: the charge conjugation operator anticommutes with the Hamiltonian
\begin{equation*}
C^{-1} H C = -H.
\end{equation*}
This reflects the fact that particle states have opposite energy to their corresponding hole partners.

\subsubsection*{Chiral symmetry}

The composition of a time reversal symmetry and a particle-hole symmetry, which we could impose to commute with each other, yields a projective unitary representation of \( (\mathbb{Z}_2,\cdot) \), characterized by the chiral operator
\begin{equation*}
  S := TC,  
\end{equation*}
which also represents \( -1 \in \mathbb{Z}_2 \). Anyway, chiral symmetry need not necessarily arise from the combination of time reversal and charge conjugation symmetries since it can exists independently even if the other two anti-unitary symmetries are broken. Generally speaking, \( S^2 = \pm 1_{\mathcal{H}} \), but by redefining \( \tilde{S} := i S = e^{i\pi/2} S \), a chiral symmetry can always be seen as a representation in which
\begin{equation*}
    S^2 = +1_{\mathcal{H}}.
\end{equation*}
In the case in which $S:=TC$, then $S^2=(TC)^2=TCTC = \pm T^2C^2$, depending on whether $T$ and $C$ commute or anticommute, with $T^2=\epsilon_T1_{\mathcal{H}}$ and $C^2=\epsilon_C1_{\mathcal{H}}$ where $\epsilon_T,\epsilon_C \in \{\pm 1\}$: hence again $S^2 = 1_{\mathcal{H}}$ up to a possible multiplication by $i$. Moreover, the chiral operator anticommutes with the Hamiltonian,
\[
S^{-1} H S = -H,
\]
so chiral symmetry is not a dynamical symmetry.

\section{Winding number for unitary matrices}
\label{appendixC}
We collect here some results on the concept of winding number, first for maps from the circle to itself and then for unitary-matrix-valued map $S^1 \to \mathrm{U}(N)$. Proofs can be found e.g.\ \cite[Appendix A]{Monaco_2023} and references therein. 

\begin{Proposition}[Winding number for a map $f : S^1 \to S^1$]\label{WFSM}
    Let $f, g : S^1 \to S^1$ be two differentiable maps. The integral
\[
w(f) := \frac{1}{2\pi i} \int_{S^1}  \frac{\partial_kf_k}{f_k}dk
\]
is integer-valued and depends only on the homotopy class of $f$; it is called the winding number of $f$. The winding number of the product is the sum of the two winding numbers,
\[
w(f g) = w(f) + w(g).
\]
Moreover, let $\iota : S^1 \to S^1$ be the involution $k \mapsto -k$. We then have
\[
w(f \circ \iota) = -w(f).
\]
\end{Proposition}

\begin{Theorem}[Winding number of unitary-matrix-valued maps]\label{thm:winding_number}
Let $\pi_1(\mathrm{U}(N))$ denote the group of (smooth) homotopy classes of (smooth) maps $S^1 \to \mathrm{U}(N)$, where
\[
\mathrm{U}(N) := \{ U \in \mathrm{M}_N(\mathbb{C}) \mid U^* = U^{-1} \},
\]
endowed with point-wise multiplication of maps as the group operation. Then there is a group isomorphism $\pi_1(\mathrm{U}(N)) \simeq \mathbb{Z}$ given by
\[
\operatorname{wn}: [U : S^1 \to \mathrm{U}(N)] \mapsto \frac{1}{2\pi i} \int_{S^1} \frac{\partial_k \det (U_k)}{\det (U_k)}dk = \frac{1}{2\pi i} \int_{S^1}  \operatorname{tr}\left( U_k^* \, \partial_k U_k \right)dk.
\]
The integer $\operatorname{wn}([U])$ is called the winding number of the homotopy class $[U]$ of the map $U : S^1 \to \mathrm{U}(N)$.
\end{Theorem}

\section*{Conflict of Interest, Funding and Data Availability}
The Authors declares no conflict of interest. 

The Authors thank the anonymous referee for the thoughtful observations on the first submitted version of the present paper.\\
D. M. and G. P. gratefully acknowledge financial support from Ministero dell'Universit\`{a} e della Ricerca (MUR, Italian Ministry of University and Research) and Next Generation EU within PRIN 2022AKRC5P ``Interacting Quantum Systems: Topological Phenomena and Effective Theories'' and within PNRR-MUR Project no.~PE0000023-NQSTI, as well as Sapienza Universit\`{a} di Roma within Progetto di Ricerca di Ateneo 2023 and 2024. \\
F. M. thanks M.L.L. for her support and encouragement in undertaking the journey into mathematics as well.

No new data were created or analyzed in this study.

\bibliographystyle{plain} 
	
\bibliography{main}

\bigskip

\address{(F.~Manzoni)
Dipartimento di Fisica e Matematica, Universit\`{a} degli Studi Roma Tre \\
Via della Vasca Navale 84, 00146 Roma (Italy)\\
\email{federico.manzoni@uniroma3.it}
}

\address{(D.~Monaco)
Dipartimento di Matematica ``Guido Castelnuovo'', Sapienza Universit\`{a} di Roma\\
Piazzale Aldo Moro 5, 00185 Roma (Italy)\\
\email{domenico.monaco@uniroma1.it}
}

\address{(G.~Peluso)
Dipartimento di Matematica ``Guido Castelnuovo'', Sapienza Universit\`{a} di Roma\\
Piazzale Aldo Moro 5, 00185 Roma (Italy)\\
\email{gabriele.peluso@uniroma1.it}
}

\end{document}